\documentclass[11pt]{article}
\usepackage{color,amsmath, amsthm, amsfonts, amssymb, dsfont, graphicx, listings, graphics, epsfig, enumerate, bbm,dsfont,setspace
}

\usepackage[margin=1in]{geometry}
\usepackage[fleqn,tbtags]{mathtools}
\usepackage[english]{babel}

\usepackage[round]{natbib} 

\title{Dual representations for systemic risk measures}
\author{\c{C}a\u{g}{\i}n Ararat\thanks{Assistant Professor, Bilkent University, Department of Industrial Engineering, Ankara, Turkey, cararat@bilkent.edu.tr.}
	\and
	Birgit Rudloff\thanks{Associate Professor, Vienna University of Economics and Business, Institute for Statistics and Mathematics, Vienna, Austria, brudloff@wu.ac.at.}
}
\date{July 31, 2019}

\addto\captionsenglish {}
\makeatletter \renewenvironment{proof}[1][\proofname] {\par\pushQED{\qed}\normalfont\topsep6\p@\@plus6\p@\relax\trivlist\item[\hskip\labelsep\bfseries#1\@addpunct{.}]\ignorespaces}{\popQED\endtrivlist\@endpefalse} \makeatother
\newtheorem{thm}{Theorem}[section]
\newtheorem{cor}[thm]{Corollary}
\newtheorem{lem}[thm]{Lemma}
\newtheorem{prop}[thm]{Proposition}
\newtheorem{assumption}[thm]{Assumption}
\newtheorem{defn}[thm]{Definition}
\theoremstyle{definition}
\newtheorem{example}[thm]{Example}
\newtheorem{rem}[thm]{Remark}

\numberwithin{equation}{section}
\newcommand{\R}{\mathbb{R}}
\newcommand{\sm}{\!\setminus\!}

\DeclareMathOperator{\cl}{cl}

\DeclareMathOperator{\interior}{int}

\DeclareMathOperator{\epi}{epi}
\DeclareMathOperator{\esssup}{ess\,sup}

\let\abs=\envert

\newcommand{\F}{\mathcal{F}}
\newcommand{\X}{\mathcal{X}}
\newcommand{\A}{\mathcal{A}}
\renewcommand{\L}{\mathcal{L}}

\newcommand{\E}{\mathbb{E}}
\newcommand{\M}{\mathcal{M}}
\renewcommand{\Pr}{\mathbb{P}}
\newcommand{\Q}{\mathbb{Q}}
\renewcommand{\a}{\alpha}
\newcommand{\s}{\text{sys}}
\renewcommand{\S}{\mathbb{S}}
\renewcommand{\H}{\mathcal{H}}

\newcommand{\is}{\text{ins}}
\newcommand{\sen}{\text{sen}}
\newcommand{\sys}{\text{sys}}

\newcommand{\of}[1]{\ensuremath{\left( #1 \right)}}
\newcommand{\cb}[1]{\ensuremath{ \left\{ #1 \right\} }}
\newcommand{\sqb}[1]{\ensuremath{ \left[ #1 \right] }}
\newcommand{\norm}[1]{\ensuremath{ \left\Vert #1 \right\Vert }}
\newcommand{\ip}[1]{\ensuremath{ \left\langle #1 \right\rangle }}

\def\prehp(#1,#2){\ensuremath{  #1 \cdot #2 }}

\begin{document}
\maketitle
\thispagestyle{empty}

\begin{abstract}
The financial crisis showed the importance of measuring, allocating and regulating systemic risk. Recently, the systemic risk measures that can be decomposed into an aggregation function and a scalar measure of risk, received a lot of attention. In this framework, capital allocations are added after aggregation and can represent bailout costs. More recently, a framework has been introduced, where institutions are supplied with capital allocations before aggregation. This yields an interpretation that is particularly useful for regulatory purposes. In each framework, the set of all feasible capital allocations leads to a multivariate risk measure. In this paper, we present dual representations for scalar systemic risk measures as well as for the corresponding multivariate risk measures concerning capital allocations. Our results cover both frameworks: aggregating after allocating and allocating after aggregation. As examples, we consider the aggregation mechanisms of the Eisenberg-Noe model as well as those of the resource allocation and network flow models.\\
\\[-5pt]
\\[-5pt]
\end{abstract}

\section{Introduction}\label{intro}

Systemic risk can be regarded as the inability of an interconnected system to function properly. In the financial mathematics community, defining, measuring and allocating systemic risk has been of increasing interest especially after the recent financial crisis. This paper is concerned with the representations and economic interpretations of some recently proposed measures of systemic risk from a convex duality point of view.

Canonically, network models are used for the analysis of systemic risk as proposed by the pioneering work of \cite{eisnoe}. In this model, the institutions of an interconnected financial system are represented by the nodes of a network and the liabilities of these institutions to each other are represented on the arcs. Under mild nondegeneracy conditions, it is proved in \cite{eisnoe} that the system can reach an equilibrium by realizing a unique clearing payment mechanism computed as the solution of a fixed point problem. The Eisenberg-Noe model is generalized in various directions since then, for instance, by taking into account illiquidity \citep{clearing}, default costs \citep{rogersveraart}, randomness in liabilities \citep{cim}, central clearing \citep{amini}, to name a few. The reader is refered to \cite{kabanovsurvey} for a survey of various clearing mechanisms considered in the literature.

More recently, several authors have considered the question of measuring systemic risk in relation to the classical framework of monetary risk measures in \cite{artzner}. The following three-step structure can be seen as a blueprint for the systemic risk measures defined in the recent literature \citep{amini, samuel, fouque, dito, cim, syst, kromer, hoffmann}.
\begin{itemize}
	\item \textbf{Aggregation function: }The aggregation function quantifies the impact that the random shocks of the system have on society by taking into account the interconnectedness of the institutions. It is a multivariate function that takes as input the random wealths (shocks) of the individual institutions and gives as output a scalar quantity that represents the impact of the financial system on society or on real economy. In the Eisenberg-Noe model, for instance, one can simply add society to the financial network as an additional node and define the value of the aggregation function as the net equity of society after clearing payments are realized. More simplistic choices of the aggregation function can consider total equities and losses, only total losses, or certain utility functions of these quantities; see \cite{cim, kromer}.
	
	\item \textbf{Acceptance set: }As the wealths of the institutions are typically subject to randomness, the aggregation function outputs a random quantity accordingly. The next step is to test these random values with respect to a criterion for riskiness, which is formalized by the notion of the acceptance set $\A$ of a monetary risk measure $\rho$. For instance, one can consider the acceptance set of the (conditional) value-at-risk at a probability level and check if the random total loss of the system is an element of this acceptance set.
	
	\item \textbf{Systemic risk measure: }The last step is to define the systemic risk measure based on the choices of the aggregation function $\Lambda$ and the acceptance set $\A$. \cite{cim} proposed the first axiomatic study for measuring systemic risk based on monetary risk measures, where the systemic risk measure is defined as
	\begin{equation}\label{scalarsyst}
	\rho^{\is}(X)=\rho(\Lambda(X))=\inf\cb{k\in\R\mid \Lambda(X)+k \in \A},
	\end{equation}
	where the argument $X$ is a $d$-dimensional random vector denoting the wealths of the institutions. In a financial network model, the value of this systemic risk measure can be interpreted as the minimum total endowment needed in order to make the equity of society acceptable. If one is interested in the individual contributions of the institutions to systemic risk, $\rho^{\is}(X)$ needs to be allocated back to these institutions. To be able to consider the \emph{measurement} and \emph{allocation} of systemic risk at the same time, the values of systemic risk measures are defined in \cite{syst} as \emph{sets of vectors} of individual capital allocations for the institutions. Hence, the systemic risk measures in \cite{syst} map into the power set of $\R^d$, that is, they are \emph{set-valued} functionals. For instance, the set-valued counterpart of $\rho^{\is}(X)$ is defined as
	\[
	R^{\is}(X) = \cb{z\in\R^d\mid \Lambda(X)+\sum_{i=1}^d z_i \in \A}.
	\]
	The risk measures $\rho^{\is}$ and $R^{\is}$ are considered \emph{insensitive} as they do not take into account the effect of the additional endowments in the aggregation procedure. Thus, they can be interpreted as bailout costs: the costs of making a system acceptable after the random shock $X$ of the system has impacted society. In contrast to this, a \emph{sensitive} version is proposed in \cite{syst} (and in \cite{fouque} as scalar functionals) where the aggregation function inputs the augmented wealths of the institutions:
	\[
	R^{\sen}(X) = \cb{z\in\R^d \mid \Lambda(X+z) \in \A}.
	\]
	In analogy to \eqref{scalarsyst} one can consider the smallest overall addition of capital
	\begin{equation}\label{ex_sens}
	\rho^{\sen}(X)=\inf\cb{\sum_{i=1}^d z_i\mid \Lambda(X+z) \in \A},
	\end{equation}
	that makes the impact of $X$ on society acceptable. But in contrast to the insensitive case, the sensitive risk measures $R^{\sen}$ and $\rho^{\sen}$ can be used for regulation: by enforcing to add the capital vector $z\in\R^d$ to the wealth of the banks as the impact on society \emph{after} capital regulation, that is $\Lambda(X+z)$, is made acceptable. They are called sensitive as they take the impact of capital regulations on the system into account.
\end{itemize}

This paper provides dual representation results for the systemic risk measures $R^\is$ and $R^\sen$ as well as for their scalarizations $\rho^{\is}$ and $\rho^{\sen}$ in terms of three types of dual variables: probability measures for each of the financial institutions, weights for each of the financial institutions, and probability measures for society. The probability measures can be interpreted as possible models governing the dynamics of the institutions/society. Each time one makes a guess for these models, a \emph{penalty} is incurred according to ``how far" these measures are from the true probability measure of the financial system. Then, the so-called \emph{systemic penalty function} (Definition~\ref{syspen}) is computed as the minimized value of this penalty over all choices of the probability measure of society. According to the dual representations, the systemic risk measures $R^\is$ and $R^\sen$ collect the capital allocation vectors whose certain weighted sums pass a threshold level determined by the systemic penalty function.

In terms of economic interpretations, a convenient feature of the dual representations is that (the objective function of) the systemic penalty function has an additive structure in which the contributions of the network topology (encoded in the conjugate function of $\Lambda$) and the choice of the regulatory criterion for riskiness (encoded in the penalty function of $\rho$) are transparent. Moreover, the first term dealing with the network topology can be regarded as a \emph{multivariate divergence functional}, for instance, a multivariate relative entropy, and it can be written in a simple analytical form in many interesting cases where the aggregation function $\Lambda$ itself, as the primal object, is defined in terms of an optimization problem. For instance, this is the case for the Eisenberg-Noe model without (Section~\ref{noccp}) and with (Section~\ref{withccp}) central clearing, as well as for the classical resource allocation and network flow models of operations research.

In the general (non-systemic) setting, dual representations for risk measures are well-studied; see \cite{fs:sf} for univariate risk measures, \cite{hh:duality} for set-valued risk measures, and \cite{multipleeligible} for scalar multivariate risk measures. It should be noted that the dual representations of the present paper do not follow as consequences of the dual representations of the general framework. This is because both the insensitive and the sensitive  systemic risk measures are defined in terms of the \emph{composition} of the univariate risk measure $\rho$ and the aggregation function $\Lambda$. In contrast to the existing duality results for general risk measures, the results of the paper ``dualize" both $\rho$ and $\Lambda$. In the special case where $\Lambda$ is a linear function, this can be achieved by the well-known Fenchel-Rockafellar theorem. On the other hand, the general case where $\Lambda$ is a concave function is less well-known. In our arguments, we use two results dealing with the general case: \citet[Theorem~2.8.10]{zalinescu}, which works under some continuity assumptions and gives a precise result for the conjugate, and the more recent \citet[Theorem~3.1]{radu}, which works under very mild conditions but identifies the conjugate up to a closure operation.

In the more traditional insensitive setting for systemic risk measures, \citet[Theorem~3]{cim} provides a dual representation for $\rho^\is$ assuming that the underlying probability space is finite and $\Lambda, \rho$ are positively homogeneous functions. Under these assumptions, $\rho^{\is}$ can be computed as the optimal value of a finite-dimensional linear optimization problem and the corresponding dual problem is regarded as a dual reprensentation for $\rho^\is$. This result is generalized by \cite{kromer} for general probability spaces, convex $\rho$ and concave $\Lambda$. It should be noted that the dual representation for $\rho^\is$ given in the present paper provides a different economic interpretation than the ones in \cite{cim,kromer}. In particular, \citet[Theorem~3]{cim} is stated in terms of sub-probability measures (sub-stochastic vectors) and the ``remaining" mass to extend such a measure to a probability measure is interpreted as a probability assigned to an artificial \emph{scenario} $\omega_0$ added to the underlying probability space. In contrast, Theorem~\ref{mainthm} and  Proposition~\ref{scalarizations} of the present paper are stated in terms of probability measures corresponding to the institutions as well as an additional probability measure corresponding to society. Note that society is considered as an additional \emph{node} in the network of institutions.

To the best of our knowledge, dual representations of systemic risk measures in the sensitive case ($R^\sen$ and $\rho^{\sen}$) have not yet been studied in the literature besides a few special cases. Among the related works, \citet[Theorem~2.10]{samuel} provides a dual representation for $\rho^\sen$ in the special case where $\A=\cb{X\mid \E\sqb{-X}\leq 0}$, that is, $\rho$ is the negative expected value. More recently, \citet[Section~3]{fairness} studies the dual representation of a type of sensitive systemic risk measure which considers random capital allocations (different from the one in the present paper) where the aggregation function $\Lambda$ is a decomposable sum of univariate utility functions and $\rho$ is the negative expected value.

The rest of the paper is organized as follows. In Section~\ref{general}, the definitions of the systemic risk measures are recalled along with some basic properties. In Section~\ref{main}, the main results of the paper are collected in Theorem~\ref{mainthm} followed by some comments on the economic interpretation of these results. The form of the dual representations under some canonical aggregation functions, including that of the Eisenberg-Noe model, are investigated in Section~\ref{examples}. A model uncertainty representation of the sensitive systemic risk measure is discussed in Section~\ref{modeluncertainty}. Finally, Section~\ref{pf}, the Appendix, is devoted to proofs.

\section{Insensitive and sensitive systemic risk measures}\label{general}

We consider an interconnected financial system with $d$ institutions. By a \emph{realized state of the system}, we mean a vector $x=(x_1,\ldots,x_d)^{\mathsf{T}}\in\R^d$, where $x_i$ denotes the wealth of institution $i$. To compare two possible states $x,z\in\R^d$, we use the componentwise ordering $\leq$ on $\R^d$; hence, $x\leq z$ if and only if $x_i \leq z_i$ for every $i\in\cb{1,\ldots,d}$. We write $\R^d_+=\cb{x\in\R^d\mid 0\leq x}$.

Given a realized state, the interconnectedness of the system is taken into account through a single quantity provided by the so-called \emph{aggregation function}. Formally speaking, this is a function $\Lambda\colon\R^d\to\R$ satisfying the following properties.
\begin{enumerate}[\bf (i)]
	\item \textbf{Increasing:} $x\leq z$ implies $\Lambda(x)\leq \Lambda(z)$ for every $x,z\in\R^d$.
	\item \textbf{Concave:} It holds $\Lambda(\gamma x + (1-\gamma) z) \geq \gamma \Lambda(x) + (1-\gamma) \Lambda(z)$ for every $x,z\in\R^d$ and $\gamma  \in [0,1]$.
	\item \textbf{Non-constant:} $\Lambda$ has at least two distinct values.
\end{enumerate}
As discussed in Section~\ref{intro}, $\Lambda(x)$ can be interpreted as the impact of the system on society given that the state of the system is $x\in\R^d$. An overall increase in the wealth of the system is anticipated to have a positive impact on society, which is reflected by the property that $\Lambda$ is increasing. Similarly, the concavity of $\Lambda$ reflects that diversification in wealth has a positive impact on society. Finally, the last condition eliminates the trivial case that $\Lambda$ is a constant, which ensures that the set $\Lambda(\R^d)\coloneqq\cb{\Lambda(x)\mid x\in\R^d}$ has an interior point.

To model the effect of a financial crisis, a catastrophic event, or any sort of uncertainty affecting the system, we assume that the state of the system is indeed a random vector $X$ on a probability space $(\Omega,\F,\Pr)$. Hence, the impact on society is realized to be $\Lambda(X(\omega))$ if the observed scenario for the uncertainty is $\omega\in\Omega$. For convenience, we assume that $X\in L_d^\infty$, where $L_d^\infty$ is the space of $d$-dimensional essentially bounded random vectors that are distinguished up to almost sure equality. Consequently, the impact on society is a univariate random variable $\Lambda(X)\in L^\infty$, where $ L^\infty = L^\infty_1$. Throughout, we call $\Lambda(X)$ the \emph{aggregate value} of the system.

The systemic risk measures we consider are defined in terms of a measure of risk for the aggregate values. To that end, we let $\rho\colon L^\infty\to\R$ be a convex monetary risk measure in the sense of \cite{artzner}. More precisely, $\rho$ satisfies the following properties. (Throughout, (in)equalities between random variables are understood in the $\Pr$-almost sure sense.)
\begin{enumerate}[\bf (i)]
	\item \textbf{Monotonicity: }$Y_1\geq Y_2$ implies $\rho(Y_1)\leq \rho(Y_2)$ for every $Y_1, Y_2\in L^\infty$.
	\item \textbf{Translativity: }It holds $\rho(Y+y)=\rho(Y)-y$ for every $Y\in L^\infty$ and $y\in\R$.
	\item \textbf{Convexity: }It holds $\rho(\gamma Y_1 + (1-\gamma)Y_2)\leq \gamma \rho(Y_1)+(1-\gamma)\rho(Y_2)$ for every $Y_1, Y_2\in L^\infty$ and $\gamma\in[0,1]$.
	\item \textbf{Fatou property: }If $(Y_n)_{n\geq 1}$ is a bounded sequence in $L^\infty$ converging to some $Y\in L^\infty$ almost surely, then $\rho(Y)\leq \liminf_{n\rightarrow\infty}\rho(Y_n)$.
\end{enumerate}
The risk measure $\rho$ is characterized by its so-called \emph{acceptance set} $\A\subseteq L^\infty$ via the following relationships
\[
\A =\cb{Y\in L^\infty \mid \rho(Y)\leq 0},\quad \quad \rho(Y)=\inf\cb{y\in\R\mid Y+y\in\A}.
\]
Hence, the aggregate value $\Lambda(X)$ is considered acceptable if $\Lambda(X)\in \A$.

As a well-definedness assumption for the systemic risk measures of interest, we will need the following, where $\interior \Lambda(\R^d)$ denotes the interior of the set $\Lambda(\R^d)$.
\begin{assumption}\label{assume}
	$\rho(0)\in \interior\Lambda(\R^d)$.
\end{assumption}

\begin{rem}Note that Assumption~\ref{assume} can be replaced with the weaker assumption that $\interior\Lambda(\R^d)$ is a nonempty set, which is already satisfied thanks to the assumption that $\Lambda$ is a non-constant function. In that case, by shifting $\Lambda$ by a constant, one can easily obtain an aggregation function that satisfies Assumption~\ref{assume}.\end{rem}

Finally, we recall the definitions of the two systemic risk measures of our interest. As in \cite{samuel, syst}, we adopt the so-called set-valued approach, namely, systemic risk is measured as the \emph{set} of all capital allocation vectors that make the system safe in the sense that the aggregate value becomes acceptable when the institutions are supplied with these capital allocations. We consider first the insensitive case, where institutions are supplied with capital allocations \emph{after} aggregation, and then consider the sensitive case, where institutions are supplied with capital allocations \emph{before} aggregation.

We start by recalling the set-valued analog of the systemic risk measure in \cite{cim}. In what follows, $2^{\R^d}$ denotes the power set of $\R^d$ including the empty set.

\begin{defn}\label{insensitive}
	\cite[Example~2.1.(i)]{syst} The insensitive systemic risk measure is the set-valued function $R^{\is}\colon L_d^\infty \to 2^{\R^d}$ defined by
	\[
	R^{\is}(X) = \cb{z\in\R^d\mid \Lambda(X)+\sum_{i=1}^d z_i \in \A}
	\]
	for every $X\in L_d^\infty$.
\end{defn}

\begin{rem}\label{rem-ins}
	Note that
	\begin{align}\label{ins-rep}
	R^{\is}(X) = \cb{z\in\R^d\mid \rho\of{\Lambda(X)+\sum_{i=1}^d z_i}\leq 0}=\cb{z\in\R^d \mid \rho^\is(X)\leq \sum_{i=1}^d z_i},
	\end{align}
	where $\rho^\is = \rho\circ\Lambda$ is the scalar systemic risk measure in \cite{cim}, see \eqref{scalarsyst}.  It follows from \eqref{ins-rep} that
	\[
	\rho^\is (X) = \inf_{z\in R^\is (X)} \sum_{i=1}^d z_i.
	\]
	Hence, $\rho^\is(X)$ and $R^\is(X)$ determine each other.
\end{rem}

As motivated in Section~\ref{intro}, a more ``sensitive" systemic risk measure can be defined by aggregating the wealths after the institutions are supplied with their capital allocations.

\begin{defn}\label{sensitive}
	\cite[Example~2.1.(ii)]{syst} The sensitive systemic risk measure is the set-valued function $R^{\sen}\colon L_d^\infty \to 2^{\R^d}$ defined by
	\[
	R^{\sen}(X) = \cb{z\in\R^d \mid \Lambda(X+z) \in \A}
	\]
	for every $X\in L_d^\infty$.
\end{defn}

\begin{rem}\label{rem-sen}
	For fixed $X\in L_d^\infty$, note that
	\begin{equation}\label{sen-rep}
	R^\sen (X) = \cb{z\in\R^d \mid \rho(\Lambda(X+z))\leq 0}=\cb{z\in\R^d\mid \rho^\is(X+z)\leq 0}.
	\end{equation}
	However, $R^\sen(X)$ cannot be recovered from $\rho^\is(X)$, in general.
\end{rem}

Let us denote by $L_d^1$ the set of all $d$-dimensional random vectors $X$ whose expectations $\E\sqb{X}\coloneqq(\E\sqb{X_1},\ldots,\E\sqb{X_d})^{\mathsf{T}}$ exist as points in $\R^d$.

\begin{defn}\label{setrmdefn}
	\cite[Definition~2.1]{hh:duality}
	For a set-valued function $R\colon L_d^\infty\to 2^{\R^d}$, consider the following properties.
	\begin{enumerate}[\bf (i)]
		\item \textbf{Monotonicity: }$X\geq Z$ implies $R(X)\supseteq R(Z)$ for every $X,Z\in L_d^\infty$.
		\item \textbf{Convexity: }It holds $R(\gamma X+(1-\gamma) Z)\supseteq \gamma R(X)+(1-\gamma)R(Z)$ for every $X,Z\in L_d^\infty$ and $\gamma\in[0,1]$.
		\item \textbf{Closedness: }The set $\cb{X\in L_d^\infty \mid z\in R(X)}$ is closed with respect to the weak$^{\ast}$ topology $\sigma(L_d^\infty, L_d^1)$ for every $z\in \R^d$.
		\item \textbf{Finiteness at zero: }It holds $R(0)\notin\cb{\emptyset,\R^d}$.
		\item \textbf{Translativity: }It holds $R(X+z)=R(X)-z$ for every $X\in L_d^\infty$ and $z\in\R^d$.
		\item \textbf{Positive homogeneity: }It holds $R(\gamma X) = \gamma R(X)\coloneqq \cb{\gamma z\mid z\in R(X)}$ for every $X\in L_d^\infty$ and $\gamma >0$.
	\end{enumerate}
\end{defn}

\begin{prop}\label{properties}
	\text{} 
	\begin{enumerate}
		\item $R^{\is}$ is a set-valued convex risk measure that is non-translative in general: it satisfies properties (i)-(iv) above.
		\item $R^{\sen}$ is a set-valued convex risk measure: it satisfies all of properties (i)-(v) above.
	\end{enumerate}
\end{prop}

The proof of Proposition~\ref{properties} is given in Section~\ref{proof2}.

\begin{rem}\label{rem-sen2}
	Let $X\in L_d^\infty$. An immediate consequence of Proposition~\ref{properties} is that $R^{\sen}(X)$ is a closed convex subset of $\R^d$ satisfying $R^{\sen}(X)=R^{\sen}(X)+\R^d_+$. Hence,  we may write $R^{\sen}(X)$ as the intersection of its supporting halfspaces
	\[
	R^{\sen}(X)=\bigcap_{w\in\R^d_+\sm\cb{0}}\cb{z\in\R^d\mid w^{\mathsf{T}}z\geq \rho^\sen_w(X)},
	\]
	where
	\begin{equation*}
	\rho^\sen_w(X) \coloneqq \inf_{z\in R^\sen(X)}w^{\mathsf{T}}z = \inf_{z\in\R^d}\cb{w^{\mathsf{T}}z\mid \Lambda(X+z)\in\A},
	\end{equation*}
	for each $w\in\R^d_+\sm\cb{0}$. In other words, $\rho^{\sen}_w$ is the \emph{scalarization} of the set-valued function $R^\sen$ in direction $w\in\R^d_+\sm\cb{0}$ and is a scalar measure of systemic risk, see \citet[Definition~3.3]{syst}. The family $(\rho^\sen_w(X))_{w\in\R^d_+\sm\cb{0}}$ determines $R^{\sen}(X)$; compare Remark~\ref{rem-ins} and Remark~\ref{rem-sen}. If one chooses $w=(1,\ldots,1)^{\mathsf{T}}\in\R^d$, then one obtains the risk measure given in \eqref{ex_sens}.
\end{rem}

We conclude this section with sufficient conditions that guarantee the positive homogeneity of the systemic risk measures; see (vi) of Definition~\ref{setrmdefn}.

\begin{prop}\label{coherence}
	Suppose that $\Lambda$ and $\rho$ are positively homogeneous, that is, $\Lambda(\gamma  x)=\gamma \Lambda(x)$ and $\rho(\gamma  X)=\gamma \rho(X)$ for every $x\in \R^d, X\in L_d^\infty, \lambda>0$. Then, $R^\is$ and $R^\sen$ are positively homogeneous.
\end{prop}

The proof of Proposition~\ref{coherence} is given in Section~\ref{proof2}.

\section{Dual representations}\label{main}

The main results of this paper provide dual representations for the insensitive ($R^{\is}$) and sensitive ($R^{\sen}$) systemic risk measures and their scalarizations $\rho^{\is}$ and $\rho^{\sen}$.

The dual representations are formulated in terms of probability measures and (weight) vectors in $\R^d$. Given two finite measures $\mu_1,\mu_2$ on $(\Omega,\F)$, we write $\mu_1\ll\mu_2$ if $\mu_1$ is absolutely continuous with respect to $\mu_2$. We denote by $\M(\Pr)$ the set of all probability measures $\Q$ on $(\Omega,\F)$ such that $\Q\ll\Pr$. In addition, we denote by $\M_d(\Pr)$ the set of all \emph{vector probability measures} $\Q=\of{\Q_1, \ldots, \Q_d}^{\mathsf{T}}$ whose components are in $\M(\Pr)$. Let $\mathbf{1}$ be the vector in $\R^d$ whose components are all equal to $1$.

Let us denote by $g$ the Legendre-Fenchel conjugate of the convex function $x\mapsto-\Lambda(-x)$, that is,
\begin{equation}\label{conjugate}
g(z) = \sup_{x\in\R^d}\of{\Lambda(x)-z^{\mathsf{T}}x }
\end{equation}
for each $z\in\R^d$. A direct consequence of the monotonicity of $\Lambda$ is that $g(z)=+\infty$ for every $z\notin\R^d_+$, hence we will only consider the values of $g$ for $z\in\R^d_+$.

In addition, since $\rho\colon L^\infty\to\R$ is a convex monetary risk measure satisfying the Fatou property, it has the dual representation
\[
\rho(Y) = \sup_{\S\in\M(\Pr)}\of{\E^{\S}\sqb{-Y}-\alpha(\S)}
\]
for every $Y\in L^{\infty}$, where $\alpha$ is the (minimal) penalty function of $\rho$ defined by
\begin{equation}\label{penalty}
\alpha(\S)\coloneqq \sup_{Y\in\A}\E^{\S}\sqb{-Y} = \sup_{Y\in L^{\infty}}\of{\E^{\S}\sqb{-Y}-\rho(Y)}
\end{equation}
for $\S\in \M(\Pr)$; see \citet[Theorem~4.33]{fs:sf}, for instance.

For two vectors $x,z\in\R^d$, their Hadamard product is defined by
\[
x\cdot z \coloneqq \of{x_1 z_1, \ldots, x_d z_d}^{\mathsf{T}}.
\]

\begin{defn}\label{syspen}
	The function $\a^{\sys}\colon\M_d(\Pr)\times\of{\R^d_+\sm\cb{0}}\to\R\cup\cb{+\infty}$ defined by
	\begin{equation*}
	\alpha^{\sys}(\Q,w)\coloneqq \inf_{\substack{ \S\in\M(\Pr)\colon\\ \forall i\colon w_i\Q_i\ll\S}}\of{\a(\S)+\E^{\S}\sqb{g\of{w\cdot \frac{d\Q}{d\S}}}}
	\end{equation*}
	for every $\Q\in\M_d(\Pr)$, $w\in\R^d_+\sm\cb{0}$ is called the systemic penalty function.
\end{defn}

In the above definition, for every $i\in\cb{1,\ldots,d}$, the condition $w_i\Q_i\ll\S$ becomes trivial when $w_i=0$ and is equivalent to $\Q_i\ll\S$ when $w_i>0$; hence, $w_i\Q_i\ll\S$ can be replaced with the condition
\[
w_i>0\;\Rightarrow\; \Q_i\ll\S
\]
equivalently. We make the convention that $w_i\frac{d\Q_i}{d\S}=0$ when $w_i=0$ although $\Q_i\ll\S$ is not required in this case. On the other hand, $g(z)\geq \Lambda(0)$ for every $z\in\R^d$, by \eqref{conjugate}. Hence, $g$ is bounded from below since $\Lambda$ is a real-valued function. These make the quantity $\E^\S\sqb{g\of{w\cdot\frac{d\Q}{d\S}}}$ in Definition~\ref{syspen} well-defined.

The following theorem summarizes the main results of the paper. Its proof is given in Section~\ref{proof}, the appendix.

\begin{thm}\label{mainthm}
	The insensitive and sensitive systemic risk measures admit the following dual representations.
	\begin{enumerate}
		\item For every $X\in L_d^\infty$,
		\begin{align*}
		R^{\is}(X)&=\bigcap_{\Q\in \M_d(\Pr), w\in\R^d_+\sm\cb{0}}\cb{z\in \R^d \mid \mathbf{1}^{\mathsf{T}}z\geq w^{\mathsf{T}}\E^{\Q}\sqb{-X}-\a^{\sys}(\Q,w)}\\
		&=\cb{z\in \R^d \mid \mathbf{1}^{\mathsf{T}}z\geq \sup_{\Q\in \M_d(\Pr), w\in\R^d_+\sm\cb{0}}\of{w^{\mathsf{T}}\E^{\Q}\sqb{-X}-\a^{\sys}(\Q,w)}}.
		\end{align*}
		\item For every $X\in L_d^\infty$,
		\begin{align*}
		R^{\sen}(X)&=\bigcap_{\Q\in \M_d(\Pr), w\in\R^d_+\sm\cb{0}}\cb{z\in \R^d \mid w^{\mathsf{T}}z\geq w^{\mathsf{T}}\E^{\Q}\sqb{-X}-\a^{\sys}(\Q,w)}\\
		&=\bigcap_{\Q\in \M_d(\Pr), w\in\R^d_+\sm\cb{0}} \Big(\E^{\Q}\sqb{-X}+\cb{z\in \R^d \mid w^{\mathsf{T}}z\geq -\a^{\sys}(\Q,w)}\Big).
		\end{align*}
	\end{enumerate}
\end{thm}

Let us comment on the economic interpretation of the above dual representations. Consider a financial network with nodes $1,\ldots,d$ denoting the institutions and society (or an external entity) is added to this network as node $0$. The dual representations can be regarded as the conservative computations of the capital allocations of the institutions in the presence of \emph{model uncertainty} and \emph{weight ambiguity} according to the following procedure.
\begin{itemize}
	\item Society is assigned a probability measure $\S$, which has the associated penalty $\a(\S)$. 
	\item Each institution $i$ is assigned a probability measure $\Q_i$ and a relative weight $w_i$ with respect to society.
	\item The distance of the network of institutions to society is computed by the \emph{multivariate $g$-divergence} of $(\Q_1, \ldots, \Q_d)$ with respect to $\S$ as follows. Each density $\frac{d\Q_i}{d\S}$ is multiplied by its associated relative weight $w_i$, and the weighted densities are used as the input of the divergence function $g$. The resulting multivariate $g$-divergence is
	\[
	\E^{\S}\sqb{g\of{w_1\frac{d\Q_1}{d\S},\ldots,w_d\frac{d\Q_d}{d\S}}},
	\]
	which can be seen as a weighted sum distance of the vector probability measure $\Q$ to the probability measure $\S$ of society. In particular, when the aggregation function $\Lambda$ is in a certain exponential form (Section~\ref{entropic} below), the multivariate $g$-divergence is a weighted sum of relative entropies with respect to $\S$.
	\item The systemic penalty function $\a^{\sys}$ is computed as the minimized sum of the multivariate $g$-divergence and the penalty function $\a$ over all possible choices of the probability measure $\S$ of society, see Definition~\ref{syspen}. This is the total penalty incurred for choosing $\Q$ and $w$ as a probabilistic model of the financial system.
	\item \textbf{Insensitive case: }To compute $R^{\is}(X)$, one computes the worst case weighted negative expectation of the wealth vector $X$ penalized by the systemic penalty function over all possible choices of the uncertain model $\Q\in\M_d(\Pr)$ and the ambigious weight vector $w\in\R^d_+\sm\cb{0}$:
	\[
	\rho^{\is}(X)=\sup_{\Q\in \M_d(\Pr), w\in\R^d_+\sm\cb{0}}\of{w^{\mathsf{T}}\E^{\Q}\sqb{-X}-\a^{\sys}(\Q,w)}.
	\]
	This quantity serves as the minimal total endowment needed for the network of institutions: every capital allocation vector $z\in\R^d$ whose sum of entries exceeds $\rho^{\is}(X)$ is considered as a feasible compensator of systemic risk, and hence, it is included in the set $R^{\is}(X)$. In particular, $R^{\is}(X)$ is a halfspace with direction vector $\mathbf{1}$.
	
	\item \textbf{Sensitive case: }To compute $R^{\sen}(X)$, one computes the negative expectation of the wealth vector $X$ penalized by the systemic penalty function. This quantity serves as a threshold for the \emph{weighted} total endowment of the institutions: a capital allocation vector $z\in\R^d$ is considered feasible with respect to the model $\Q\in\M_d(\Pr)$ and weight vector $w\in\R^d_+\sm\cb{0}$ if its\emph{ weighted sum} exceeds its corresponding threshold, that is, if
	\[
	w^{\mathsf{T}}z\geq w^{\mathsf{T}}\E^{\Q}\sqb{-X}-\a^{\sys}(\Q,w).
	\]
	Finally, a capital allocation vector $z\in\R^d$ is considered as a feasible compensator of systemic risk if it is feasible with respect to \emph{all} possible choices of the model $\Q$ and weight vector $w$.
\end{itemize}

\begin{rem}\label{nomodeluncertainty}
	Let us consider the special case where the risk measure $\rho$ for the aggregate values is $\rho(Y)=\E\sqb{-Y}$ for every $Y\in L^\infty$. In this case, we have $\a(\S)=0$ if and only if $\S=\Pr$, and $\a(\S)=+\infty$ otherwise. In view of the above economic interpretations, this choice of $\rho$ corresponds precisely to the case where there is no uncertainty about the probability measure of \emph{society}. In particular, the systemic penalty function reduces simply to the multivariate $g$-divergence with respect to the true probability measure $\Pr$, that is,
	\[
	\a^\sys(\Q,w)=\E\sqb{g\of{w\cdot\frac{d\Q}{d\Pr}}}
	\]
	for every $\Q\in\M_d(\Pr)$ and $w\in\R^d_+\sm\cb{0}$. Nevertheless, the model uncertainty (as well as the weight ambiguity) associated to the \emph{banks} remains in the picture since one has still to calculate the intersections over different choices of $(\Q,w)$ in Theorem~\ref{mainthm}. This observation can be seen as a justification of the interpretation that the aggregation function $\Lambda$ is a society-related quantity: as $\rho$ is used to test whether $\Lambda(X)$ is acceptable, a simplistic risk-neutral choice of $\rho$ eliminates only the part of model uncertainty coming from society. Similarly, in the general case where an arbitrary risk measure $\rho$ is used, the quantity $\a(\S)$ is the dual object associated to the acceptability of $\Lambda(X)$, which justifies the interpretation of $\S$ as society's probability measure.
\end{rem}

As a follow-up on Theorem~\ref{mainthm}, we state below the dual representations of the so-called \emph{scalarizations} of the insensitive and sensitive systemic risk measures. Recall from Remark~\ref{rem-ins} that
\[
\rho^\is(X) = \inf_{z\in R^\is(X)} \mathbf{1}^{\mathsf{T}}z
\]
for every $X\in L_d^\infty$, where $\rho^\is = \rho\circ \Lambda$. Hence, $\rho^\is$ is the \emph{scalarization} of the set-valued function $R^\is$ in the direction $\mathbf{1}$. From \eqref{ins-rep}, it is clear that the values of $R^\is$ are halfspaces with normal direction $\mathbf{1}$. Hence, the scalarizations of $R^\is$ in different directions yield trivial values, that is, for every $X\in L_d^\infty$,
\[
\inf_{z\in R^\is(X)}w^{\mathsf{T}}z = - \infty
\]
provided that $w\in\R^d_+\sm\cb{0}$ is not of the form $w=\lambda \mathbf{1}$ for some $\lambda>0$. On the other hand, this is not the case for $R^\sen$ as its values are not halfspaces in general. 

For the sensitive case, recall from Remark~\ref{rem-sen2} the scalarizations
\[
\rho^\sen_w(X) = \inf_{z\in R^\sen(X)}w^{\mathsf{T}}z = \inf_{z\in\R^d}\cb{w^{\mathsf{T}}z\mid \Lambda(X+z)\in\A},
\]
for $X\in L_d^\infty$ and $w\in\R^d_+\sm\cb{0}$. One such scalarization can be used as a scalar measure of systemic risk if one can fix \emph{a priori} a weight vector $w\in\R^d_+\sm\cb{0}$ which implies a ranking of the importance of the institutions. While $\rho^\sen_w$ is a monotone convex functional, it has the following form of translativity that depends on the choice of $w$: for every $X\in L_d^\infty, z\in\R^d$,
\[
\rho^\sen_w(X+z) = \rho_w^\sen(X)-w^{\mathsf{T}}z.
\]

Comparing Remark~\ref{rem-sen2} and Theorem~\ref{mainthm}, one can ask if we have equality in
\begin{equation}\label{conj2}
\rho^\sen_w(X)\stackrel{?}{=} \sup_{\Q\in\M_d(\Pr)}\of{w^{\mathsf{T}}\E^\Q\sqb{-X}-\a^\s(\Q,w)}.
\end{equation}
However, $\rho^\sen_w$ might fail to be a weak* lower semicontinuous function in general, even though $R^\sen$ is a closed set-valued function; see \citet[page~92]{setoptsurv} for a discussion about the lower semicontinuity of the scalarizations of set-valued functions. Therefore, one can only expect to have a dual representation for $\rho^\sen_w$ when it is assumed to be weak* lower semicontinuous. Furthermore, $\a^\s(\Q,\cdot)$ may fail to be positively homogeneous in general while $w\mapsto \rho^\sen_w(X)$ and $w\mapsto w^{\mathsf{T}}\E^\Q\sqb{-X}$ are positively homogeneous. For this reason, $\a^\sys(\Q,\cdot)$ should be replaced in \eqref{conj2} with a positively homogeneous alternative. In Proposition~\ref{scalarizations}, under a technical condition, we provide such a version of \eqref{conj2} in which equality is achieved.

As a preparation for Proposition~\ref{scalarizations}, we introduce some additional notation. Let us denote by $L_d^1$ the linear space of all integrable $d$-dimensional random vectors (distinguished up to almost sure equality). For $p\in\cb{1,+\infty}$, let us also define the cone $L^p_{d,+} = \cb{U\in L_d^p \mid \Pr\cb{U\geq 0}=1}$; if $d=1$, then we write $L^p = L_1^p$, $L^p_+=L^p_{1,+}$. We denote by $\rho^\ast$ the conjugate function of $\rho$ defined by
\[
\rho^\ast(V)\coloneqq \sup_{Y\in L^\infty}\of{\E\sqb{VY}-\rho(Y)}
\]
for each $V\in L^1$.

Let us consider the function $m$ on $L_d^1$ defined by
\begin{equation}\label{mbar}
m(U)\coloneqq\inf_{V\in L_+^1}\cb{\E\sqb{Vg\of{\frac{U}{V}}1_{\cb{V>0}}}+\E\sqb{V}\rho^\ast\of{\frac{-V}{\E\sqb{V}}}\mid \Pr\cb{V=0,U\neq 0}=0}
\end{equation}
for each $U\in L_{d,+}^1$, where $\E\sqb{V}\rho^\ast(\frac{-V}{\E\sqb{V}})=0$ is understood when $V\equiv 0$; and $m(U)\coloneqq+\infty$ for $U\notin L_{d,+}^1$. We denote by $\cl m$ the closure of $m$, that is, $\cl m$ is the unique function on $L_d^1$ whose epigraph is the closure of the epigraph of $m$; see Section~\ref{proof} for the definition of epigraph. The function $m$ is an essential element of Proposition~\ref{scalarizations} and it gives rise to the systemic penalty function under additional assumptions. The role of $m$ is discussed further in the proof of Proposition~\ref{scalarizations} in Section~\ref{proof}. For the time being, we need it to state the dual representations of scalarizations.

\begin{prop}\label{scalarizations}
	The scalarizations of the insensitive and sensitive systemic risk measures admit the following dual representations.
	\begin{enumerate}[1.]
		\item For every $X\in L_d^\infty$,
		\[
		\rho^\is (X) = \sup_{\Q\in\M_d(\Pr), w\in\R^d_+\sm\cb{0}}\of{w^{\mathsf{T}}\E^\Q\sqb{-X}-\a^\s(\Q,w)}.
		\]
		\item Let $w\in\R^d_+\sm\cb{0}$ and assume that $\rho^\sen_w$ is weak* lower semicontinuous. Then, for every $X\in L_d^\infty$,
		\begin{equation}\label{mrep}
		\rho_w^\sen (X) = \sup_{\Q\in\M_d(\Pr)}\of{w^{\mathsf{T}}\E^\Q\sqb{-X}-(\cl m)\of{w\cdot\frac{d\Q}{d\Pr}}}.
		\end{equation}
		Moreover, if $m$ is lower semicontinuous, then
		\begin{equation}\label{scsimple}
		\rho_w^\sen (X) = \sup_{\Q\in\M_d(\Pr)}\of{w^{\mathsf{T}}\E^\Q\sqb{-X}-\tilde{\a}^\s(\Q,w)},
		\end{equation}
		where
		$\tilde{\a}^\s(\Q,\cdot)$ is the positively homogeneous function generated by $\a^\s(\Q,\cdot)$ (see \citet[Chapter~5]{rockafellar}), namely,
		\begin{equation}\label{atilde}
		\tilde{\a}^\s(\Q,w)\coloneqq \inf_{\lambda>0}\frac{\a^\s(\Q,\lambda w)}{{\lambda}}.
		\end{equation}
		In particular, if there exist $\hat{X}\in L_d^\infty$ and a (weak*) neighborhood $A$ of $\Lambda(\hat{X})$ such that $A\subseteq\A$, then $m$ is lower semicontinuous and thus \eqref{scsimple} holds.
	\end{enumerate}
	Consequently, $\rho^\is$ and $\rho^\sen_{\mathbf{1}}$ do not coincide, in general.
\end{prop}

The second part of Proposition~\ref{scalarizations} gives rise to an alternative dual representation for $R^\sen$ under the stated assumptions, which is given in the following corollary.

\begin{cor}\label{altdual}
	For every $w\in\R^d_+\sm\cb{0}$, suppose that $\rho_w^\sen$ is a weak* lower semicontinuous function. In addition, assume that $m$ is lower semicontinuous. Then, for every $X\in L_d^\infty$,
	\[
	R^\sen(X)=\bigcap_{\Q\in\M_d(\Pr),w\in\R^d_+\sm\cb{0}}\cb{z\in\R^d\mid w^\mathsf{T}z\geq w^{\mathsf{T}}\E^\Q\sqb{-X}-\tilde{\a}^\sys(\Q,w)},
	\]
	where $\tilde{a}^\sys$ is defined as in \eqref{atilde}.
\end{cor}
\begin{proof}
	The result is an immediate consequence of Proposition~\ref{scalarizations} and Remark~\ref{rem-sen2}.
\end{proof}
\begin{rem}\label{relativeweight}
	Corollary~\ref{altdual} can be used to justify the interpretation of the dual variable $w\in\R^d_+$ as a vector of \emph{relative} weights. It can be assumed that the \emph{absolute} weight of society is $w_0=1$ and the weights $w_1,\ldots,w_d$ of the institutions are relative to this value of $w_0$. In an alternative formulation, one can work with absolute weights $\bar{w}_0>0,\bar{w}\in\R^d_+\sm\cb{0}$ for both the institutions and society. Then, it follows from \eqref{atilde} that
	\begin{equation}\label{abs}
	\tilde{\a}^\sys(\Q,\bar{w})=\inf_{\substack{\bar{w}_0>0,\\ \S\in\M(\Pr)\colon\\ \forall i\colon w_i\Q_i\ll\S }}\of{\bar{\a}\of{\bar{w}_0\frac{d\S}{d\Pr}}+\bar{w}_0\E^\S\sqb{g\of{\frac{\bar{w}}{\bar{w}_0}\cdot\frac{d\Q}{d\S}}}},
	\end{equation}
	where $\bar{\a}\of{\bar{w}_0\frac{d\S}{d\Pr}}\coloneqq\sup_{Y\in\A}\bar{w}_0\E^\S[-Y]$ extends the definition of $\a$ in \eqref{penalty} for the finite measure $\bar{w}_0\S$. In this formulation, for each $i\in\cb{1,\ldots,d}$, the fraction $\frac{\bar{w_i}}{\bar{w_0}}$ is the relative weight of institution $i$ with respect to society. Theorem~\ref{mainthm} suggests that the sensitive systemic risk measure $R^\sen$ is scale-free in the sense that only relative weights matter for the calculation of $R^\sen$. Hence, it is enough to consider the case $\bar{w}_0=w_0=1$ and write down the dual representation in terms of the relative weight vector $\frac{\bar{w}}{\bar{w}_0}=w$.
	These observations are also in line with the fact that the systemic penalty function $\a^\sys$ is not positively homogeneous in the relative weight variable: $\a^\sys(\Q,\lambda w)$ and $\lambda \a^\sys(\Q, w)$ do not coincide, in general ($\lambda>0$). On the other hand, the expression in the infimum in \eqref{abs} is positively homogeneous as a function of the absolute weight vector $(\bar{w}_0,\bar{w_1},\ldots,\bar{w}_d)\in\R^{d+1}$.
\end{rem}

We end this section with the dual representation of the systemic risk measures when they are guaranteed to be positively homogeneous by virtue of Proposition~\ref{coherence}.

\begin{cor}\label{coherence-rep}
	Suppose that $\Lambda$ and $\rho$ are positively homogeneous. Then, there exists a nonempty closed convex set $\mathcal{Z}\subseteq \R_+^d$ such that
	\[
	g(z)=\begin{cases}
	0 & \text{ if }z\in\mathcal{Z},\\ +\infty &\text{ else.}\end{cases}
	\]
	Besides, there exists a convex set $\mathcal{S}\subseteq \M(\Pr)$ of probability measures such that
	\[
	\a(\S) = \begin{cases}0 &\text{ if }\S\in\mathcal{S},\\ +\infty & \text{ else}.\end{cases}
	\]
	Let
	\[
	\mathcal{D}\coloneqq\cb{(\Q,w)\in \M_d(\Pr)\times\R^d_+\sm\cb{0} \mid \exists\S\in\mathcal{S}\colon\of{\Pr\cb{w\cdot \frac{d\Q}{d\S}\in\mathcal{Z}}=1\ \wedge \ \forall i\colon w_i\Q_i\ll \S}}.
	\]
	Then, the insensitive and sensitive systemic risk measures admit the following dual representations.
	\begin{enumerate}[1.]
		\item For every $X\in L_d^\infty$,
		\[
		R^\is(X) = \bigcap_{(\Q,w)\in \mathcal{D}} \cb{z\in\R^d \mid \mathbf{1}^{\mathsf{T}}z\geq w^{\mathsf{T}}\E^{\Q}\sqb{-X}}
		\]
		and
		\[
		\rho^\is (X) = \sup_{(\Q,w)\in \mathcal{D}}w^{\mathsf{T}}\E^\Q\sqb{-X}.
		\]
		
		\item For every $X\in L_d^\infty$,
		\begin{align*}
		R^\sen(X) &= \bigcap_{(\Q,w)\in \mathcal{D}} \cb{z\in\R^d \mid w^{\mathsf{T}}z\geq w^{\mathsf{T}}\E^{\Q}\sqb{-X}}\\
		&=\bigcap_{(\Q,w)\in \mathcal{D}} \Big(\E^{\Q}\sqb{-X}+\cb{z\in\R^d \mid w^{\mathsf{T}}z\geq 0}\Big)
		\end{align*}
		where
		\[
		\mathcal{Q}^w\coloneqq \cb{\Q\in\M_d(\Pr)\mid (\Q,\lambda w)\in\mathcal{D}\text{ for some }\lambda> 0}.
		\]
	\end{enumerate}
\end{cor}

\begin{proof}
	The existence of the set $\mathcal{Z}$ is due to the following well-known facts from convex analysis; see \citet[Theorem~13.2]{rockafellar}, for instance: a positively homogeneous proper closed convex function is the support function of a nonempty closed convex set, and the conjugate of this function is the convex indicator function of the set. The existence of the set $\mathcal{S}$ is by the dual representations of coherent risk measures; see \citet[Corollary~4.37]{fs:sf}. From Definition~\ref{syspen}, it follows that $\a^\sys(\Q,w)=0$ if $(\Q,w)\in\mathcal{D}$ and $\a^\sys(\Q,w)=+\infty$ otherwise. The rest follows from Theorem~\ref{mainthm}.
\end{proof}

\section{Examples}\label{examples}

According to Theorem~\ref{mainthm}, to be able to specify the dual representation of the insensitive and sensitive systemic risk measures, one needs to compute the penalty function $\a$ of the underlying monetary risk measure $\rho$ as well as the multivariate $g$-divergence $\E^\S[g(w\cdot \frac{d\Q}{d\S})]$ for dual probability measures $\S, \Q$ and weight vector $w$. As the penalty functions of some canonical risk measures (for instance, average value-at-risk, entropic risk measure, optimized certainty equivalents) are quite well known, we focus on the computation of multivariate $g$-divergences here. In the following subsections, we consider some canonical examples of aggregation functions proposed in the systemic risk literature.

\subsection{Total profit-loss model}

One of the simplest ways to quantify the impact of the system on society is to aggregate all profits and losses in the system \cite[Example~1]{cim}. This amounts to setting 
\[
\Lambda(x) = \sum_{i=1}^d x_i.
\]
for every realized state $x\in\R^d$. In this case, it is clear from Definition~\ref{insensitive} and Definition~\ref{sensitive} that $R^\is = R^\sen$.

An elementary calculation using \eqref{conjugate} yields
\[
g(z) = \begin{cases}
0 & \text{ if }z=\mathbf{1},\\
+\infty& \text{ else},
\end{cases}
\]
for every $z\in\R^d$. Hence, given dual variables $\Q\in \M_d(\Pr)$, $\S\in \M(\Pr)$, $w\in \R^d_+\sm\cb{0}$ with $w_i\Q_i\ll\S$ for each $i\in\cb{1,\ldots,d}$, we have
\[
\E^\S\sqb{g\of{w\cdot \frac{d\Q}{d\S}}}=
\begin{cases}
0 & \text{ if }w=\mathbf{1}, \Q_i=\S\text{ for every }i\in\cb{1,\ldots,d},\\
+\infty&\text{ else}.
\end{cases}
\]
As a result, once a measure $\S$ is chosen for society, the only plausible choice of the measure $\Q_i$ of institution $i$ is $\S$, and any other choice would yield infinite $g$-divergence. Therefore,
\[
\a^\sys(\Q,w)=
\begin{cases}
\a(\S)&\text{ if }w=\mathbf{1}, \Q_1=\ldots=\Q_d=\S\text{ for some }\S \in \M(\Pr),\\
+\infty & \text{ else},
\end{cases}
\]
and one obtains
\[
R^\is (X) = R^\sen (X) = \cb{z\in\R^d\mid \mathbf{1}^{\mathsf{T}}z\geq -\inf_{\S\in\M(\Pr)}\of{\sum_{i=1}^d\E^\S\sqb{X_i}+\a(\S)}}
\]
for every $X\in L_d^\infty$.

\subsection{Total loss model}

The previous example of aggregation function can be modified so as to take into account only the losses in the system \cite[Example~2]{cim}, that is, we can define
\[
\Lambda(x) = -\sum_{i=1}^d x_i^-
\]
for every $x\in \R^d$. In this case, the insensitive and sensitive systemic risk measures no longer coincide.

The conjugate function for the total loss model is given by
\[
g(z) = \begin{cases}
0 & \text{ if }z_i\in [0,1]\text{ for every }i\in\cb{1,\ldots,d},\\
+\infty& \text{ else},
\end{cases}
\]
for every $z\in \R^d$. Hence, given $\Q\in \M_d(\Pr)$, $\S\in \M(\Pr)$, $w\in \R^d_+\sm\cb{0}$ with $w_i\Q_i\ll\S$ for each $i\in\cb{1,\ldots,d}$,
\[
\E^\S\sqb{g\of{w\cdot \frac{d\Q}{d\S}}}=
\begin{cases}
0 & \text{ if } \Pr \cb{w_i\frac{d\Q_i}{d\S}\leq 1}=1\text{ for every }i\in\cb{1,\ldots,d},\\
+\infty&\text{ else}.
\end{cases}
\]
Therefore, the systemic penalty function can be given as
\[
\a^\sys(\Q,w)=\inf_{\S\in\M(\Pr)}\cb{\a(\S)\mid  w_i\Q_i\ll\S,\ \Pr\cb{w_i\frac{d\Q_i}{d\S}\leq 1}=1\text{ for every }i\in\cb{1,\ldots,d}}.
\]

\subsection{Entropic model}\label{entropic}

As an example of a strictly concave aggregation function, let us suppose that $\Lambda$ aggregates the profits and losses through an exponential utility function \cite[Section~5.1(iii)]{syst}, namely,
\[
\Lambda(x) = -\sum_{i=1}^d e^{-x_i-1}
\]
for every $x\in \R^d$. Then, for every $z\in\R^d_+$,
\[
g(z) = \sum_{i=1}^d z_i \log(z_i),
\]
where $\log(0)\coloneqq -\infty$ and $0\log(0)\coloneqq 0$ by convention. Hence, for every $\Q\in \M_d(\Pr)$, $\S\in \M(\Pr)$, $w\in \R^d_+\sm\cb{0}$ with $w_i\Q_i\ll\S$ for each $i\in\cb{1,\ldots,d}$, the $g$-divergence is given by
\[
\E^\S\sqb{g\of{w\cdot \frac{d\Q}{d\S}}}=\sum_{i=1}^d \H\of{w_i\Q_i\Vert \S},
\]
where $\H\of{w_i\Q_i\Vert \S}$ is the \emph{relative entropy} of the finite measure $w_i\Q_i$ with respect to society's probability measure $\S$, that is,
\[
\H\of{w_{i}\Q_i\Vert \S} \coloneqq\E^\S\sqb{w_i\frac{d\Q_i}{d\S}\log\of{w_i\frac{ d\Q_i}{d\S}}}.
\]
Since $\H(w_i\Q_i\Vert \S)=w_i\H(\Q_i \Vert\S)+w_i\log(w_i)$, one can also write
\[
\E^\S\sqb{g\of{w\cdot \frac{d\Q}{d\S}}}=\sum_{i=1}^d w_i \H\of{\Q_i\Vert\S}+ \sum_{i=1}^d w_i\log(w_i).
\]
Hence, the systemic penalty function has the form
\begin{align*}
\a^\s(\Q,w) = \inf_{\substack{ \S\in\M(\Pr)\colon\\ \forall i\colon w_i\Q_i\ll\S}}\of{\a(\S)+\sum_{i=1}^d \H\of{w_i \Q_i\Vert\S}}= \inf_{\substack{ \S\in\M(\Pr)\colon\\ \forall i\colon w_i\Q_i\ll\S}}\of{\a(\S)+\sum_{i=1}^d w_i\H\of{\Q_i\Vert\S}}+c(w),
\end{align*}
where 
$c(w)\coloneqq \sum_{i=1}^d w_i\log(w_i)$.

Finally, we consider a special case where the underlying monetary risk measure $\rho$ is the entropic risk measure, that is,
\[
\rho(Y) =\log \E \sqb{e^{- Y}}
\]
for every $Y\in L^\infty$. In this case, the penalty function of $\rho$ is also a relative entropy:
\[
\a(\S) = \H\of{\S\Vert \Pr}.
\]
As a result, the systemic penalty function becomes
\[
\a^\s(\Q,w) = \inf_{\substack{ \S\in\M(\Pr)\colon\\ \forall i\colon w_i\Q_i\ll\S}} \of{ \H\of{\S\Vert\Pr}+\sum_{i=1}^d w_i \H\of{\Q_i\Vert\S}}+c(w).
\]
As relative entropy is a commonly used quantification of distance between probability measures, this form of the systemic penalty function provides a geometric insight to the economic interpretations discussed in Section~\ref{main}. Indeed, the sum $\H\of{\S\Vert\Pr}+\sum_{i=1}^d w_i \H\of{\Q_i\Vert\S}$ can be seen as the weighted sum distance of the vector probability measure $\Q$ to the physical measure $\Pr$ while passing through the probability measure $\S$ of society: as a first step, one measures the distance from each $\Q_i$ to $\S$ as $\H\of{\Q_i\Vert\S}$, and computes their weighted sum $\sum_{i=1}^d w_i \H\of{\Q_i\Vert\S}$. Then, this weighted sum is added to the distance $\H\of{\S\Vert\Pr}$ of $\S$ to $\Pr$, which gives the total distance of $\Q$ to $\Pr$ via $\S$. Finally, the systemic penalty function looks for the minimum possible distance of $\Q$ to $\Pr$ (via $\S$) over all choices of $\S\in\M(\Pr)$ with $w_i\Q_i\ll\S$ for every $i\in\cb{1,\ldots,d}$.

\subsection{Eisenberg-Noe model}\label{noccp}

The previous three examples provide general rules for aggregating the wealths of the institutions. As these rules ignore the precise structure of the financial system, they would be useful in systemic risk measurement, for instance, in the absence of detailed information about interbank liabilities.

In this subsection, we consider the network model of \cite{eisnoe}, where the financial institutions (typically banks) are modeled as the nodes of a network and the liabilities between the institutitons are represented on the arcs. As in \cite{syst}, we will add society as an additional node to the network and define the aggregation function as the net equity of society after clearing payments are realized based on the liabilities.

Let us recall the description of the model. We consider a financial network with nodes $0,1,\ldots, d$, where nodes $1,\ldots,d$ denote the institutions and node $0$ denotes society. For an arc $(i,j)$ with $i,j\in\cb{0,1\ldots,d}$, let us denote by $\ell_{ij}\geq 0$ the \emph{nominal liability} of node $i$ to node $j$. We make the following assumptions.
\begin{enumerate}[\bf (i)]
	\item Society has no liabilities, that is, $\ell_{0i}=0$ for every $i\in\cb{1,\ldots,d}$.
	\item Every institution has nonzero liability to society, that is, $\ell_{i0}>0$ for every $i\in \cb{1,\ldots,d}$.
	\item Self-liabilities are ignored, that is, $\ell_{ii}=0$ for every $i\in\cb{0, 1,\ldots,d}$.
\end{enumerate}
For an arc $(i,j)$ with $i\neq 0$, the corresponding \emph{relative liability} is defined as
\[
a_{ij}\coloneqq 
\frac{\ell_{ij}}{\bar{p}_i}, 
\]
where $\bar{p}_i \coloneqq \sum_{j=0}^{d}\ell_{ij}>0$ is the \emph{total liability} of institution $i$.

Given a realized state $x\in\R^d$, a vector $p(x)\coloneqq(p_1(x),\ldots, p_d(x))^{\mathsf{T}}\in\R^d_+$ is called a \emph{clearing payment vector} for the system if it solves the fixed point problem
\[
p_i(x) =\min\cb{\bar{p}_i, x_i + \sum_{j=1}^{d} a_{ji}p_j(x)},\quad i\in\cb{1,\ldots,d}.
\]
In this case, the payment $p_i(x)$ of an institution $i$ at clearing must be equal either to the total liability of $i$ (no default) or else to the total income of $i$ coming from other institutions as well as its realized wealth (default). Clearly, every clearing payment vector $p=p(x)$ is a feasible solution of the linear programming problem
\begin{align*}\label{P}
& \text{maximize }\; \sum_{i=1}^d a_{i0}p_i  \tag{$P(x)$}\\
& \text{subject to }\; p_{i}\leq x_i + \sum_{j=1}^d a_{ji} p_j,\quad i\in\cb{1,\ldots,d},  \\
&       \; \;\;\;\;\quad\quad \quad \; \;    p_i\in [0, \bar{p}_i],\quad i\in\cb{1,\ldots,d}. 
\end{align*}
Let us denote by $\Lambda(x)$ the optimal value of problem \eqref{P}. Note that this problem is either infeasible, in which case we set $\Lambda(x)=-\infty$, or else it has a finite optimal value $\Lambda(x)\in [0,\bar{\bar{p}}]$, where $\bar{\bar{p}}\coloneqq \sum_{i=0}^d a_{i0}\bar{p}_i$. Let us denote by $\X$ the set of all $x\in\R^d$ for which \eqref{P} is feasible. Clearly, $\R^d_+\subseteq \X$. In fact, only the case $x\in \R^d_+$ is considered by \cite{eisnoe} and it is shown in \citet[Lemma~4]{eisnoe} that every optimal solution of \eqref{P} is a clearing payment vector for the system. We note here that the same result holds for every $x\in \X$ since the objective function is strictly increasing with respect to the payment $p_i$ of each institution $i$.

Therefore, if $x\in\X$, then the optimal value $\Lambda(x)$ is the equity of society after clearing payments are realized, and if $x\notin \X$, then we have $\Lambda(x)=-\infty$ in which case there is no clearing payment vector. Hence, we set $\Lambda$ to be the aggregation function for the Eisenberg-Noe model as it quantifies the impact of the financial network on society.

It is easy to check that $\Lambda$ is increasing, concave and non-constant. Hence, it satisfies the definition of an aggregation function except that it may take the value $-\infty$. Nevertheless, by Remark~\ref{eisnoe-ext} below, we are able to apply Theorem~\ref{mainthm} to this choice of $\Lambda$. In Proposition~\ref{noccp-comp} below, we provide a simple expression for the conjugate function $g$ defined by \eqref{conjugate}.

\begin{prop}\label{noccp-comp}
	For $z\in\R^d_+$, one has
	\[
	g(z) = \sum_{i=1}^d c_i(z)^+,
	\]
	where
	\[
	c_i(z) = \sum_{j=0}^d \ell_{ij}(z_j -z_i),
	\]
	and $z_0\coloneqq 1$. Consequently, for every $\Q\in\M_d(\Pr)$, $\S\in \M(\Pr)$, $w\in \R^d_+\sm\cb{0}$ with $w_i\Q_i\ll\S$ for each $i\in\cb{1,\ldots,d}$,
	\[
	\E^\S\sqb{g\of{w\cdot \frac{d\Q}{d\S}}}=\sum_{i=1}^d \E^\S\sqb{c_i\of{w\cdot \frac{d\Q}{d\S}}^+}= \sum_{i=1}^d \E^\S\sqb{ \of{\sum_{j=0}^d \ell_{ij}\of{w_j\frac{d\Q_j}{d\S} -w_i\frac{d\Q_i}{d\S}}}^+},
	\]
	where $w_0\coloneqq 1, \Q_0\coloneqq \S$.
\end{prop}

\begin{proof}
	Let $z\in\R^d_+$. We have
	\begin{align*}
	g(z) &= \sup_{x\in\R^d}\of{\Lambda(x)-z^{\mathsf{T}}x}\\
	&= \sup_{p_i\in [0, \bar{p}_i],\; i\in\cb{1,\ldots,d}}\cb{\sum_{i=1}^d a_{i0}p_i-\inf_{x\in\R^d}\cb{z^{\mathsf{T}}x\mid x_i \geq p_i-\sum_{j=1}^d a_{ji} p_j,i\in\cb{1,\ldots,d}}}\\
	&=\sup_{p_i\in [0, \bar{p}_i],\; i\in\cb{1,\ldots,d}}\cb{\sum_{i=1}^d a_{i0}p_i-\sum_{i=1}^d z_i\of{ p_i-\sum_{j=1}^d a_{ji} p_j}}\\
	&=\sup_{p_i\in [0, \bar{p}_i],\; i\in\cb{1,\ldots,d}}\sum_{i=1}^d \of{a_{i0}+\sum_{j=1}^d a_{ij}z_j -z_i}p_i\\
	&=\sum_{i=1}^d c_i(z)^+
	\end{align*}
	since
	\[
	c_i (z)  = \sum_{j=0}^d \ell_{ij}\of{z_j-z_i} = \bar{p}_i \of{a_{i0}+\sum_{j=1}^d a_{ij}z_j -z_i}.
	\]
	Hence, the last statement follows.
\end{proof}

Therefore, for the multivariate $g$-divergence of the Eisenberg-Noe model, the contribution of institution $i$ is computed as follows. The difference between the weighted density $w_j\frac{d\Q_j}{d\S}$ of institution $j$ and the weighted density $w_i\frac{d\Q_i}{d\S}$ of institution $i$ is computed and this difference is multiplied by the corresponding liability $\ell_{ij}\geq 0$. The positive part of the sum of these (weighted) differences over all $j\neq i$ is the (random) measurement of the \emph{incompatibility} of $\Q_i, w_i$ for institution $i$ given the choices of $\Q_j, w_j$ for institutions $j\neq i$ as well as the choice of $\S$ for society. Finally, the expected value of this measurement gives the contribution of institution $i$ to the $g$-divergence. 

\begin{rem}\label{eisnoe-ext}
	Note that the aggregation function $\Lambda$ in this example takes the value $-\infty$ on $\R^d\sm\X$, which is not allowed in the general framework of Section~\ref{general}. In particular, $\Lambda(X)\in L^\infty$ may no longer hold true. Nevertheless, Definition~\ref{insensitive} and Definition~\ref{sensitive} of the systemic risk measures still make sense with the usual acceptance set $\A\subseteq L^\infty$ of a monetary risk measure $\rho\colon L^\infty\to\R$. One just obtains $R^\is(X)=\emptyset$ if  $\Lambda(X)\notin L^\infty$. Equivalently, one can extend $\rho$ to random variables of the form $\tilde{Z}=Z1_{F}-\infty 1_{\Omega\;\sm\; F}$ with $Z\in L^\infty$ and $F\in\F$ ($1_F$ denotes the stochastic indicator function of $F$) by
	\[
	\rho(\tilde{Z})=\begin{cases}
	\rho(Z) &\text{if }\Pr(F)=1,\\
	+\infty &\text{if }\Pr(F)<1,
	\end{cases}
	\]
	and then define $R^\is$ and $R^\sen$ by \eqref{ins-rep} and \eqref{sen-rep}. Naturally, this extended definition yields $\rho^\is(X)=\rho(\Lambda(X))=+\infty$ and $R^\is(X)=\emptyset$ if $\Pr\cb{\Lambda(X)\in\R}<1$. In other words, the insensitive systemic risk measure provides no capital allocation vectors in this case. However, with the sensitive systemic risk measure $R^\sen$, it is always possible to find a nonempty set of capital allocation vectors. Indeed, it is easy to check that, for every $X\in L_d^\infty$, the vector $\bar{z}\in\R^d$ defined by $z_i=\norm{X_i^-}_{\infty}$ for each $i\in\cb{1,\ldots,d}$ yields $\Lambda(X+\bar{z})\in L^\infty$ (as $X+\bar{z}\geq  0$), and moreover, one can find $z\in\R^d$ with $\Lambda(X+\bar{z}+z)\in\A$ so that $\bar{z}+z\in R^\sen(X)$. Finally, with the extended definition, $R^\sen$ still has the dual representation in Theorem~\ref{mainthm} with minor and obvious changes in the proof in Section~\ref{proof} and the dual representation of $R^\is$ in Theorem~\ref{mainthm} holds for $X$ with $\Lambda(X)\in L^\infty$, else $R^\is(X)=\emptyset$.
\end{rem}

\subsection{Eisenberg-Noe model with central clearing}\label{withccp}

When a central clearing counterparty (CCP) is introduced to the financial system, all liabilities between the institutions are realized through the CCP, which results in a star-shaped structure in the modified network. On the other hand, the institutions still have their liabilities to society. In this subsection, we consider the modified Eisenberg-Noe model with the CCP and society and show that the $g$-divergence in this model can be written in a similar way as in the model without the CCP.

Let us consider again the Eisenberg-Noe model without the CCP where the liabilities $\ell_{ij}$, $i,j\in\cb{0,1,\ldots,d}$, satisfy the three assumptions of the previous subsection. We add the CCP to the network as node $d+1$ and compute the liabilities between the CCP and institution $i\in\cb{1,\ldots,d}$ by
\[
\ell_{i (d+1)} \coloneqq \of{\sum_{j=1}^d \ell_{ij}-\sum_{j=1}^d \ell_{ji}}^+,\quad \ell_{(d+1)i} \coloneqq \of{\sum_{j=1}^d \ell_{ij}-\sum_{j=1}^d \ell_{ji}}^-.
\]
In other words, if the \emph{net} interbank liability of institution $i$ is positive in the original network, then this amount is set as the liability of institution $i$ to the CCP; otherwise, the absolute value of this amount is set as liability of the CCP to institution $i$. Once the liabilities of/to the CCP are set, the liabilities on the arcs $(i,j)$ with $i,j\in\cb{1,\ldots,d}$ are all set to zero but the liability $\ell_{i0}>0$ of institution $i$ to society remains the same.

In the modified network, a given realized state $x$ has $d+1$ components, that is, $x=(x_1,\ldots,x_{d+1})^{\mathsf{T}}$, and the defining fixed point problem of a clearing payment vector $p(x)=(p_1(x),\ldots,p_{d+1}(x))^{\mathsf{T}} \in\R^{d+1}_+$ can be written as
\begin{align}
& p_i(x) = \min\cb{\ell_{i(d+1)}+\ell_{i0},x_i + p_{d+1}(x)\frac{\ell_{(d+1)i}}{\sum_{j=1}^d\ell_{(d+1)j}}},\quad i\in\cb{1,\ldots,d}, \label{firstpart} \\
& p_{d+1}(x) = \min\cb{\sum_{i=1}^d \ell_{(d+1)i}, x_{d+1}+\sum_{i=1}^d p_i(x)\frac{\ell_{i(d+1)}}{\ell_{i(d+1)}+\ell_{i0}}}. \label{secondpart}
\end{align}
The corresponding linear programming problem becomes
\begin{align*}\label{tP}
&\text{maximize }\; \sum_{i=1}^d \frac{\ell_{i0}}{\ell_{i0}+\ell_{i(d+1)}}p_i  \tag{$\tilde{P}(x)$}\\
&\text{subject to }\; p_{i}\leq x_i + \frac{\ell_{(d+1)i}}{\sum_{j=1}^d\ell_{(d+1)j}}p_{d+1},\quad i\in\cb{1,\ldots,d},  \\
& \; \;\;\;\;\quad\quad \quad \; \;  p_{d+1}\leq x_{d+1}+\sum_{i=1}^d \frac{\ell_{i(d+1)}}{\ell_{i(d+1)}+\ell_{i0}}p_i,\\
&      \; \;\;\;\;\quad\quad \quad \; \;    p_i\in [0, \ell_{i(d+1)}+\ell_{i0}],\quad i\in\cb{1,\ldots,d},\\
&      \; \;\;\;\;\quad\quad \quad \; \; p_{d+1}\in \sqb{0,\; \sum_{i=1}^d \ell_{(d+1)i}}.
\end{align*}
Let us denote by $\tilde{\Lambda}(x)$ the optimal value of problem \eqref{tP} and by $\tilde{\X}$ the set of all $x\in\R^{d+1}$ for which \eqref{tP} is feasible. As in the original network, if $x\notin\tilde{\X}$, then we have $\tilde{\Lambda}(x)=-\infty$ and there exists no clearing payment vectors. On the other hand, if $x\in\tilde{\X}$, then \eqref{tP} has a finite optimal value $\tilde{\Lambda}(x)$. However, as the objective function does not depend on the payment $p_{d+1}$ of the CCP, an optimal solution of \eqref{tP} may fail to be a clearing payment vector. In particular, \citet[Lemma~4]{eisnoe} does not apply here. Nevertheless, any clearing payment vector is a solution of \eqref{tP}, and we will show in Proposition~\ref{CCPresult} that, for feasible \eqref{tP}, one can always find an optimal solution that is also a clearing payment vector.

\begin{prop}\label{CCPresult}
	Suppose $x\in \tilde{\X}$. Then \eqref{tP} has an optimal solution $p(x)\in\R^{d+1}_+$ that is also a clearing payment vector. Moreover, the optimal value $\tilde{\Lambda}(x)$ equals the equity of society after clearing payments are realized under any such solution of \eqref{tP}.
\end{prop}

\begin{proof}
	
	
	Let $p\in\R^{d+1}_+$ be an optimal solution of \eqref{tP}, which exists as \eqref{tP} is a feasible bounded linear programming problem by supposition. Let us define $p(x)\in\R^{d+1}_+$ by
	\begin{align*}
	&p_i(x)\coloneqq p_i,\quad i\in\cb{1,\ldots,d},\\
	&p_{d+1}(x)\coloneqq x_{d+1}+\sum_{i=1}^d \frac{\ell_{i(d+1)}}{\ell_{i(d+1)}+\ell_{i0}}p_i.
	\end{align*}
	Note that $p_{d+1}(x)\geq p_{d+1}\geq 0$. On the other hand, $p(x)$ satisfies the first part of the fixed point problem, namely, the system of equations in \eqref{firstpart}. This is due to the fact that the objective function has a strictly positive coefficient for $p_i(x)$ for each $i\in\cb{1,\ldots,d}$ and the conclusion can be checked in the same way as in the proof of \citet[Lemma~4]{eisnoe}. Hence, it is clear from \eqref{firstpart}, \eqref{secondpart} that $p(x)$ is a clearing payment vector. Therefore, $p(x)$ is also a feasible solution of \eqref{tP}. Finally, the objective function values of $p(x)$ and $p$ coincide. Therefore, $p(x)$ is an optimal solution of \eqref{tP}. The second statement follows from the optimality of $p(x)$.
\end{proof}

With Proposition~\ref{CCPresult}, the computations of the conjugate function $\tilde{g}$ and the corresponding multivariate $\tilde{g}$-divergence function can be seen as a special case of the computations in the original model in Section~\ref{noccp}.

\begin{cor}
	For every $z\in\R^{d+1}_+$,
	\[
	\tilde{g}(z) = \sum_{i=1}^{d}\sqb{\ell_{i0}(1 - z_i)+\ell_{i(d+1)}(z_{d+1}-z_i)}^+ + \of{\sum_{i=1}^d \ell_{(d+1)i}(z_i - z_{d+1})}^+.
	\]
	Consequently, for every $\Q\in\M_{d+1}(\Pr)$, $\S\in \M(\Pr)$, $w\in \R^{d+1}_+\sm\cb{0}$ with $w_i\Q_i\ll\S$ for each $i\in\cb{1,\ldots,d+1}$,
	\begin{align*}
	\E^\S\sqb{\tilde{g}\of{w\cdot \frac{d\Q}{d\S}}} &= \sum_{i=1}^{d}\E\sqb{\ell_{i0}\of{1 - w_i\frac{d\Q_i}{d\S}}+\ell_{i(d+1)}\of{w_{d+1}\frac{\Q_{d+1}}{d\S}-w_i\frac{d\Q_i}{d\S}}}^+ \\
	&\quad + \E\sqb{\sum_{i=1}^d \ell_{(d+1)i}\of{w_i\frac{d\Q_i}{d\S} - w_{d+1}\frac{d\Q_{d+1}}{d\S}}}^+
	\end{align*}
\end{cor}

\begin{proof}
	This is a special case of Proposition~\ref{noccp-comp} for a network with $d+1$ nodes and society.
\end{proof}

\subsection{Resource allocation model}\label{resource}

The resource allocation problem is a classical operations research problem where the aim is to allocate $d$ limited resources for $m$ different tasks so as to maximize the profit made from these tasks. In the systemic risk context, this problem is discussed in \cite{cim} as well.

To be precise, let us fix the problem data $p\in\R_+^m, A\in\R_+^{d\times m}$ where $p_j$ denotes the unit profit made from task $j$ and $A_{ij}$ denotes the utilization rate of resource $i$ by task $j$, for each $i\in\cb{1,\ldots,d}, j\in\cb{1,\ldots,m}$. We also denote by $u\in\R^m$ an allocation vector where $u_j$ quantifies the production in task $j\in\cb{1,\ldots,m}$.  In addition, the realized state of the system is a vector $x\in\R^d$ where $x_i$ denotes the capacity of resource $i\in\cb{1,\ldots,d}$. Then, the aggregation function is defined as the profit made from allocating the capacities optimally for the tasks, namely, $\Lambda(x)$ is the optimal value of the following linear programming problem.
\begin{align*}
&\text{maximize }\; p^{\mathsf{T}}u  \\
&\text{subject to }\; Au\leq x, \\
& \; \;\;\;\;\quad\quad \quad \; \; \; u\geq 0.
\end{align*}
As in Remark~\ref{eisnoe-ext} of the Eisenberg-Noe model, it can be argued that the infeasible case $\Lambda(x)=-\infty$ creates no problems for the application of the general duality result Theorem~\ref{mainthm}. The following proposition provides the special form of the multivariate $g$-divergence and the systemic penalty function.

\begin{prop}\label{resource-prop}
	For every $z\in\R^d_+$,
	\[
	g(z) = \begin{cases}
	0 & \text{ if }A^{\mathsf{T}}z\geq p,\\
	+\infty & \text{ else}.
	\end{cases}
	\]
	Consequently, for every $\Q\in\M_d(\Pr)$, $\S\in \M(\Pr)$, $w\in \R^d_+\sm\cb{0}$ with $w_i\Q_i\ll\S$ for each $i\in\cb{1,\ldots,d}$,
	\[
	\E^\S\sqb{g\of{w\cdot\frac{d\Q}{d\S}}}=\begin{cases}
	0 & \text{ if }\Pr\cb{A^{\mathsf{T}}\of{w\cdot\frac{d\Q}{d\S}}\geq p}=1,\\
	+\infty& \text{ else},
	\end{cases}
	\]
	and
	\[
	\a^\s(\Q,w) = \inf_{\S\in\M(\Pr)}\cb{\a(\S)\mid \Pr\cb{A^{\mathsf{T}}\of{w\cdot\frac{d\Q}{d\S}}\geq p}=1,\ w_i\Q_i\ll\S\text{ for every }i\in\cb{1,\ldots,d}}.
	\]
\end{prop}

\begin{proof}
	Note that
	\begin{align*}
	g(z) &= \sup_{x\in \R^d} \of{\Lambda(x)-z^{\mathsf{T}}x}\\
	&= \sup_{u\in\R^m_+}\of{p^{\mathsf{T}}u-\inf_{\cb{x\in\R^d\mid x\geq Au}}z^{\mathsf{T}}x}\\
	&= \sup_{u\in\R^m_+}\of{p^{\mathsf{T}}u-z^{\mathsf{T}}Au}\\
	&= \sup_{u\in\R^m_+}(p-A^{\mathsf{T}}z)^{\mathsf{T}}u,
	\end{align*}
	which is the value of the support function of the cone $\R^m_+$ in the direction $p-A^{\mathsf{T}}z$. Hence,
	\[
	g(z) = \begin{cases}
	0 & \text{ if }A^{\mathsf{T}}z-p\in \R^m_+,\\
	+\infty & \text{ else},
	\end{cases}
	\]
	which proves the first claim. The rest follows directly from the definitions of the multivariate $g$-divergence and the systemic penalty function.
\end{proof}

In light of Proposition~\ref{resource-prop}, let us comment on the interpretation of the dual variables. To each resource $i$, we assign a probability measure $\Q_i$ and a weight $w_i$. In addition, we assign a probability measure $\S$ to the economy. Then, the weighted density $w_i\frac{d\Q_i}{d\S}$ can be seen as the unit profit made from using resource $i$. Given $\S$, we say that the choices of $\Q, w$ are compatible with $\S$ if, for each $j$, the total profit made out of a unit activity in task $j$ exceeds the original unit profit for task $j$ (with probability one), that is, if
\[
\sum_{i=1}^d A_{ij}w_i\frac{d\Q_i}{d\S}\geq p_j.
\]

\subsection{Network flow model}\label{network}



The maximum flow problem aims to maximize the total flow from a source node to a sink node in a capacitated network \citep{harris, maxflowhist}. In the systemic risk context, this problem is discussed in \cite{cim} as well.

Let us formally recall the problem. We consider a network $(\mathcal{N},\mathcal{E})$, where $\mathcal{N}$ is the set of nodes and $\mathcal{E}\subseteq \mathcal{N}\times\mathcal{N}$ is the nonempty set of arcs with $d\coloneqq\abs{\mathcal{E}}$. On this network, each arc $(a,b)$ has some capacity $x_{(a,b)}\in\R$ for carrying flow. Then, $x=(x_{(a,b)})_{(a,b)\in\mathcal{E}}\in\R^{d}$ is a realized state of this system. We are interested in maximizing the flow from a fixed source node $s\in \mathcal{N}$ to a fixed sink node $t\in \mathcal{N}\sm\cb{s}$.

In this example, we will consider the so-called \emph{path formulation} of the maximum flow problem as a linear programming problem. To that end, let us recall that a (simple) path $p$ is a finite sequence of arcs where no node is visited more than once. Let $P$ be the set of all paths starting from $s$ and ending at $t$, and let $m\coloneqq\abs{P}$. For each $p\in P$, we will denote by $u_p\in\R$ a flow carried over path $p$. Then, the aggregation function is defined as the maximum total flow carried over the paths in $P$, that is,  $\Lambda(x)$ is the optimal value of the following linear programming problem.
\begin{align*}
&\text{maximize}\; \sum_{p\in P} u_p\\
&\text{subject to } \sum_{\cb{p\in P\mid(a,b)\in p}}u_p\leq x_{(a,b)},\; (a,b)\in \mathcal{E}.\\
\end{align*}

As in Remark~\ref{eisnoe-ext}, it can be argued that the infeasible case $\Lambda(x)=-\infty$ creates no problems for the application of the general duality results. The following proposition provides the special form of the multivariate $g$-divergence and the systemic penalty function.

\begin{prop}\label{maxflow-prop}
	For every $z=(z_{(a,b)})_{(a,b)\in\mathcal{E}}\in\R^d_+$,
	\[
	g(z) =\begin{cases}
	0 & \text{ if }\sum_{(a,b)\in p}z_{(a,b)}= 1 \text{ for every }p\in P,\\
	+\infty & \text{ else}.
	\end{cases}
	\
	\]
	Consequently, for every $\Q=(\Q_{(a,b)})_{(a,b)\in\mathcal{E}}\in\M_d(\Pr)$, $\S\in\M(\Pr), w=(w_{(a,b)})_{(a,b)\in\mathcal{E}}\in\R^d_+\sm\cb{0}$ with $w_{(a,b)}\Q_{(a,b)}\ll\S$ for every $(a,b)\in\mathcal{E}$,
	\begin{align*}
	\E^{\S}\sqb{g\of{w\cdot \frac{d\Q}{d\S}}} = \begin{cases}
	0 & \text{ if }\Pr\cb{\sum_{(a,b)\in p}w_{(a,b)}\frac{d\Q_{(a,b)}}{d\S} = 1}=1 \text{ for every }p\in P,\\
	+\infty & \text{ else},
	\end{cases}
	\end{align*}
	and
	\begin{align*}
	\a^\s(\Q,w)= \inf_{\substack{\S\in \M(\Pr)\colon\\ \forall(a,b)\in\mathcal{E}\colon w_{(a,b)}\Q_{(a,b)}\ll\S}}\cb{\a(\S)\mid \Pr\cb{\sum_{(a,b)\in p}w_{(a,b)}\frac{d\Q_{(a,b)}}{d\S}= 1}=1\text{ for every }p\in P }.
	\end{align*}
\end{prop}

\begin{proof}
	Let $z\in\R^d_+$. We have
	\begin{align*}
	g(z) &= \sup_{x\in\R^d}\of{\Lambda(x)-z^{\mathsf{T}}x}\\
	&= \sup_{u\in\R^m}\of{\sum_{p\in P}u_p - \inf_{x\in\R^d}\cb{z^{\mathsf{T}}x\mid \sum_{\cb{p\in P\mid(a,b)\in p}}u_p\leq x_{(a,b)}\text{ for every }(a,b)\in \mathcal{E}}}\\
	&= \sup_{u\in\R^m}\sum_{p\in P}\of{1-\sum_{(a,b)\in p}z_{(a,b)}}u_p\\
	&=\begin{cases}
	0 & \text{ if }\sum_{(a,b)\in p}z_{(a,b)}= 1 \text{ for every }p\in P,\\
	+\infty & \text{ else}.
	\end{cases} 
	\end{align*}
	The rest follows directly from the definitions of the multivariate $g$-divergence and the systemic penalty function.
\end{proof}

Note that the sensitive systemic risk measure $R^{\sen}$ provides a quantification of the risk resulting from a random shock $X=(X_{(a,b)})_{(a,b)\in\mathcal{E}}$ that affects the capacities of the arcs. In light of Proposition~\ref{maxflow-prop}, we assign a probability measure $\Q_{(a,b)}$ and a weight $w_{(a,b)}$ to each arc $(a,b)$. In addition, we assign a probability measure $\S$ to the (possibly hypothetical) arc $(s,t)$, which provides a direct connection from the source to the sink. We also assume that the weight of this arc is $w_{(s,t)}=1$. Then, the weighted density $w_{(a,b)}\frac{d\Q_{(a,b)}}{d\S}$ can be seen as the unit cost of carrying a unit flow on arc $(a,b)$ and the unit cost of carrying a unit flow on arc $(s,t)$ is $1$. Therefore, given $\S$, we say that the choices of $\Q, w$ are compatible with $\S$ if, for each path $p\in P$, the total cost of carrying a unit flow along $p$ coincides with the cost of carrying a unit flow directly from the source to the sink (with probability one), that is, if
\[
\sum_{(a,b)\in p} w_{(a,b)}\frac{d\Q_{(a,b)}}{d\S}=1=w_{(a,b)}\frac{d\S}{d\S}.
\]

\section{Model uncertainty interpretation}\label{modeluncertainty}

We finish the main part of the paper by pointing out an observation that bridges the sensitive systemic risk measure $R^\sen$ with the so-called \emph{multivariate utility-based shortfall risk measures} of recent interest in the literature.

\subsection{Multivariate shortfall risk measure}

As the aggregation function $\Lambda\colon\R^d\to\R$ is assumed to be increasing and concave, it can be seen as a multivariate utility function. Motivated by its univariate counterpart introduced in \cite{scalarshortfall}, a multivariate shortfall risk measures with respect to $\Lambda$ can be defined as follows.

\begin{defn}\label{shortfall}
	Let $\lambda^0\in - \interior\Lambda(\R^d)$. The set-valued function $R(\cdot;\Pr,\lambda^0)\colon L_d^\infty\to 2^{\R^d}$ defined by
	\[
	R(X; \Pr, \lambda^0) = \cb{z\in \R^d \mid \E\sqb{-\Lambda(X+z)}\leq \lambda^0}
	\]
	for $X\in L_d^\infty$ is called the shortfall risk measure with threshold level $\lambda^0$ and model $\Pr$.
\end{defn}

The financial interpretation of the shortfall risk measure is that, for a multivariate financial position $X$, it collects the set of all deterministic porfolios $z\in\R^d$ for which the \emph{expected loss} of $X+z$ does not exceed the fixed threshold level $\lambda^0$. (The use of $\Pr$ in the notation $R(X;\Pr,\lambda^0)$ will become clear when this risk measure is considered under model uncertainty in Section~\ref{uncertainty} below.) Such risk measures based on multivariate utility functions have been studied recently in \cite{sdrm, samuel}.

The shortfall risk measure $R(\cdot;\Pr,\lambda^0)$ defined above is an example of a sensitive systemic risk measure where the risk measure for aggregate values is chosen to be a shifted expectation, namely,
\begin{equation}\label{neg-exp}
\rho(Y)=\E\sqb{-Y}-\lambda^0
\end{equation}
for every $Y\in L^\infty$. A direct application of Theorem~\ref{mainthm} yields the following dual representation. As this representation suggests when compared to Theorem~\ref{mainthm}, using $R(\cdot;\Pr,\lambda^0)$ as a systemic risk measure amounts to assuming that the probability measure (the model) for society is known with certainty and is equal to $\Pr$.

\begin{prop}\label{shortfall-dual}
	In the setting of Definition~\ref{shortfall}, it holds
	\[
	R(X;\Pr,\lambda^0) =  \bigcap_{\Q\in \M_d(\Pr), w\in\R^d_+\sm\cb{0}}\cb{z\in \R^d \mid w^{\mathsf{T}}z\geq w^{\mathsf{T}}\E^{\Q}\sqb{-X}-\lambda^0-\E\sqb{g\of{w\cdot\frac{d\Q}{d\Pr}}}}
	\]
	for every $X\in L_d^\infty$.
\end{prop}

\begin{proof}
	This is immediate from Theorem~\ref{mainthm} once we realize that the penalty function of the risk measure defined in \eqref{neg-exp} is given by
	\[
	\a(\S) = \begin{cases}
	\lambda^0 & \text{ if }\S = \Pr, \\
	+\infty & \text{ else},
	\end{cases}
	\]
	for $\S\in \M(\Pr)$.
\end{proof}

\begin{example}
	Consider the Eisenberg-Noe model without central clearing as in Section~\ref{noccp}. In this case, the shortfall risk measure $R(\cdot;\Pr.\lambda^0)$ takes the form
	\begin{align*}
	&R(X;\Pr,\lambda^0)\\
	&=\cb{z\in\R^d\mid \E\sqb{\sup\cb{\sum_{i=1}^d a_{i0}p_i\mid p_i\leq X_i+z_i+\sum_{j=1}^d a_{ji}p_j,p_i\in[0,\bar{p}_i],i\in\cb{1,\ldots,d}}}\geq -\lambda^0}\\
	&=\bigcap_{\Q\in \M_d(\Pr), w\in\R^d_+\sm\cb{0}}\cb{z\in \R^d \mid w^{\mathsf{T}}z\geq w^{\mathsf{T}}\E^{\Q}\sqb{-X}\negthinspace -\negthinspace\lambda^0\negthinspace -\negthinspace \sum_{i=1}^d \E\sqb{ \of{\sum_{j=0}^d \ell_{ij}\of{w_j\frac{d\Q_j}{d\Pr}\negthinspace -\negthinspace w_i\frac{d\Q_i}{d\Pr}}}^+}}
	\end{align*}
	for $X\in L_d^\infty$, by Proposition~\ref{noccp-comp} and Proposition~\ref{shortfall-dual}.
\end{example}

\subsection{A model uncertainty representation}\label{uncertainty}

The connection between shortfall risk measures and the sensitive systemic risk measure $R^\sen$ can be exploited further by rearranging the order of intersections/suprema in the dual representation given by Theorem~\ref{mainthm}. In what follows, we show that \emph{any} sensitive systemic risk measure $R^\sen$ (i.e. for an arbitrary choice of $\rho$) can be seen as a shortfall risk measure under \emph{model uncertainty}.

\begin{prop}\label{modeluncertaintyrep}
	Suppose that $\Lambda(\R^d)=\R$. It holds
	\[
	R^\sen(X) = \bigcap_{\cb{\S\in\M^e(\Pr)\;\mid\; \a(\S)\in\R}}R(X;\S,\a(\S))=\bigcap_{\cb{\S\in\M^e(\Pr)\;\mid\; \a(\S)\in\R}}\cb{z\in\R^d\mid \E^\S\sqb{-\Lambda(X+z)}\leq \a(\S)}
	\]
	for every $X\in L_d^\infty$.
\end{prop}

In other words, regardless of the choice of $\rho$, the sensitive systemic risk measure can always be seen as a shortfall risk measure subject to an uncertainty in the probability measure $\S$ of society. To measure systemic risk, one makes a conservative computation (intersection) of the shortfall risk over all sensible choices of $\S$. Moreover, in each of the shortfall risk measure $R(\cdot;\S,\a(\S))$, the penalty $\a(\S)$ for choosing $\S$ serves as a maximum allowable expected loss under $\S$.

\begin{proof}[Proof of Proposition~\ref{modeluncertaintyrep}]
	For fixed $X\in L_d^\infty$, we have
	\begin{align*}
	R^\sen(X) &= \bigcap_{\Q\in \M_d(\Pr), w\in\R^d_+\sm\cb{0}}\cb{z\in \R^d \mid w^{\mathsf{T}}z\geq w^{\mathsf{T}}\E^{\Q}\sqb{-X}-\a^{\sys}(\Q,w)}\\
	&=\bigcap_{\S\in\M(\Pr)}\bigcap_{\substack{\Q\in \M_d(\Pr), w\in\R^d_+\sm\cb{0}\colon \\ \forall i\colon w_i\Q_i\ll\S}}\cb{z\in \R^d \mid w^{\mathsf{T}}z\geq w^{\mathsf{T}}\E^{\Q}\sqb{-X}-\a(\S)-\E^\S\sqb{g\of{w\cdot\frac{d\Q}{d\S}}}}\\
	&=\bigcap_{\S\in\M(\Pr)}\bigcap_{\substack{\Q\in \M_d(\Pr), w\in\R^d_+\sm\cb{0}\colon \\ \forall i\colon \Q_i\ll\S}}\cb{z\in \R^d \mid w^{\mathsf{T}}z\geq w^{\mathsf{T}}\E^{\Q}\sqb{-X}-\a(\S)-\E^\S\sqb{g\of{w\cdot\frac{d\Q}{d\S}}}}\\
	&=\bigcap_{\S\in\M(\Pr)}\bigcap_{\Q\in \M_d(\S), w\in\R^d_+\sm\cb{0}}\cb{z\in \R^d \mid w^{\mathsf{T}}z\geq w^{\mathsf{T}}\E^{\Q}\sqb{-X}-\a(\S)-\E^\S\sqb{g\of{w\cdot\frac{d\Q}{d\S}}}}\\
	&=\bigcap_{\cb{\S\in\M(\Pr)\;\mid\; \a(\S)\in -\Lambda(\R^d)}}R(X;\S,\a(\S)).
	\end{align*}
	The following arguments make the above computation valid. The first two equalities are by Theorem~\ref{mainthm}. The third equality follows from the basic observation that the halfspace inside the intersections is not affected by the choice of $\Q_i\in\M_d(\Pr)$ whenever $w_i=0$; hence we may impose $\Q_i\ll\S$ in this case as well. The fourth equality is trivial. The fifth equality follows since the inner intersection in the penultimate line is the dual representation of the shortfall risk measure with threshold level $\a(\S)\in \R$ and model $\S$; see Proposition~\ref{shortfall-dual}. Here, we need to exclude the cases where $\S\in\M(\Pr)$ is such that $\a(\S)=+\infty$. In such cases, the inner intersection in the penultimate line simply gives $\R^d$, which does not change the outer intersection in the same line. Hence, the result follows.
\end{proof}

\section{Proofs}\label{pf}

\subsection{Proofs of the results in Section~\ref{general}}\label{proof2}

\begin{proof}[Proof of Proposition~\ref{properties}]
	\text{}
	\begin{enumerate}
		\item Let $X,Z\in L_d^\infty$. To show monotonicity, suppose that $X\geq Z$. Since $\Lambda$ is increasing and $\rho$ is monotone, it holds $\rho(\Lambda(X))\leq \rho(\Lambda(Z))$. By \eqref{ins-rep}, it follows that $R^\is (X)\supseteq R^\is (Z)$. To show convexity, let $x\in  R^\is(X)$, $z\in R^\is(Z)$ and $\gamma\in[0,1]$. By \eqref{ins-rep}, it holds $\rho(\Lambda(X))\leq \sum_{i=1}^d x_i$ and $\rho(\Lambda(Z))\leq \sum_{i=1}^d z_i$. Then, 
		\begin{align*}
		\rho(\Lambda(\gamma X+(1-\gamma)Z))&\leq \rho\of{\gamma \Lambda(X)+(1-\gamma)\Lambda(Z)}\\
		&\leq \gamma\rho(\Lambda(X))+(1-\gamma)\rho(\Lambda(Z))\\
		&\leq \sum_{i=1}^d \of{\gamma x_i + (1-\gamma) z_i},
		\end{align*}
		where the first inequality follows by the concavity of $\Lambda$ and the monotonicity of $\rho$, the second inequality follows by the convexity of $\rho$, and the last inequality is by the supposition. Therefore, $\gamma x + (1-\gamma) z \in R^{\is}(\gamma X + (1-\gamma) Z)$ and convexity follows. To show closedness, let $z\in\R^d$. As a result of the convexity of $R^\is$, the  set $\L_z\coloneqq\cb{X\in L_d^\infty\mid z\in R^{\is}(X)}$ is convex. Hence, by \citet[Lemma~A.65]{fs:sf}, it suffices to show that $\mathcal{L}_{r,z}\coloneqq \cb{X\in \L_z\mid \|X\|_{\infty}\leq r}$ is closed in $L_d^1$ for every $r> 0$. (Here, $\|X\|_{\infty}\coloneqq\esssup\abs{X}$ is the essential supremum norm on $L_d^\infty$ with respect to some fixed norm $\abs{\cdot}$ on $\R^d$.) To that end, let $z\in\R^d$, $r>0$, and $(X^n)_{n\geq 1}$ be a sequence in $\L_{r,z}$ converging to some $X\in L_d^1$ in $L_d^1$. Then, there exists a subsequence $(X^{n_k})_{k\geq 1}$ converging to $X$ almost surely. Since
		\[
		\abs{X}\leq \abs{X^{n_k}-X}+\abs{X^{n_k}}\leq  \abs{X^{n_k}-X}+r
		\]
		for every $k\geq 1$, it follows that $\|X\|_\infty\leq r$. On the other hand, $(\Lambda(X^{n_k}))_{k\geq 1}$ converges to $\Lambda(X)$ almost surely since $\Lambda$ is a continuous function as a finite concave function on $\R^d$. As $(\Lambda(X^{n_k}))_{k\geq 1}$ is also a bounded sequence in $L^\infty$, by the Fatou property of $\rho$,
		\[
		\rho(\Lambda(X))\leq \liminf_{n\rightarrow\infty}\rho(\Lambda(X^{n_k}))\leq \sum_{i=1}^d z_i.
		\]
		so that $z\in R^{\is}(X)$. Hence, $X\in \L_{r,z}$ and closedness follows. Finiteness at zero is trivial from \eqref{ins-rep} since $\rho(\Lambda(X))\in \R$.
		\item Let $X,Z\in L_d^\infty$. To show monotonicity, suppose that $X\geq Z$. Since $\Lambda$ is increasing and $\rho$ is monotone, it holds $\rho(\Lambda(X+z))\leq \rho(\Lambda(Z+z))$ for every $z\in\R^d$. By \eqref{sen-rep}, it follows that $R^{\sen}(X)\supseteq R^{\sen}(Z)$. To show convexity, let $x\in R^{\sen}(X)$, $z\in R^{\sen}(Z)$ and $\gamma\in[0,1]$. By \eqref{sen-rep}, it holds $\rho(\Lambda(X+x))\leq 0$ and $\rho(\Lambda(Z+z))\leq 0$. Similar to the proof for the insensitive case,
		\begin{align*}
		\rho(\Lambda(\gamma X+(1-\gamma) Z + \gamma x +(1-\gamma) z))\leq \gamma \rho(\Lambda(X+x))+(1-\gamma)\rho(\Lambda(Z+z))\leq 0.
		\end{align*}
		Hence, $\gamma x + (1-\gamma) z \in R^\sen (\gamma X + (1-\gamma)Z)$ and convexity follows. To show closedness, similar to the proof for the insensitive case, it suffices to show that the set $\{X\in L_d^\infty\mid z\in R^\sen(X),\;$ $ \|X\|_\infty\leq r\}$ is closed in $L_d^1$ for arbitrarily fixed $r>0$ and $z\in \R^d$. Let $(X^n)_{n\geq 1}$ be a sequence in this set that converges to some $X\in L_d^1$ in $L_d^1$. Using the Fatou property of $\rho$ as above, it can be checked that $\rho(\Lambda(X+z))\leq \liminf_{n\rightarrow\infty}\rho(\Lambda(X^{n_k}+z))\leq 0$ for a subsequence $(X^{n_k})_{k\geq 1}$. Hence, $z\in R^\sen (X)$ and closedness follows. To show finiteness at zero, note that
		\[
		R^\sen (0)=\cb{z\in\R^d\mid \rho(\Lambda(z))\leq 0}=\cb{z\in\R^d\mid \rho(0)\leq \Lambda(z)}=\Lambda^{-1}([\rho(0),+\infty)),
		\]
		where the first equality is by \eqref{sen-rep} and the second equality is by the translativity of $\rho$. Since $\rho(0)\in \interior\Lambda(\R^d)$ by Assumption~\ref{assume}, it follows that $R^{\sen}(0)\notin\cb{\emptyset,\R^d}$. Finally, translativity follows since
		\begin{align*}
		R^{\sen}(X+z)&=\cb{x\in\R^d\mid\rho(\Lambda(X+z+x))\leq 0}\\
		&=\cb{x\in\R^d\mid \rho(\Lambda(X+x))\leq 0}-z\\
		&=R^{\sen}(X)-z
		\end{align*}
		for every $z\in \R^d$.
	\end{enumerate}
\end{proof}

\begin{proof}[Proof of Proposition~\ref{coherence}]
	We have
	\begin{align*}
	R^\is(\gamma  X) &=\cb{z\in \R^d\mid \rho\of{\Lambda(\gamma  X)+\sum_{i=1}^d z_i}\leq 0 }\\
	&=\cb{z\in \R^d\mid \rho\of{\gamma \Lambda(X)+\sum_{i=1}^d z_i}\leq 0 }\\
	&=\cb{z\in \R^d\mid \gamma \rho\of{\Lambda(X)+\sum_{i=1}^d \frac{z_i}{\gamma }}\leq 0 }\\
	&=\cb{z\in \R^d\mid \rho\of{\Lambda(X)+\sum_{i=1}^d \frac{z_i}{\gamma }}\leq 0 }\\
	&=\gamma  \cb{u\in \R^d\mid \rho\of{\Lambda(X)+\sum_{i=1}^d u_i}\leq 0 }\\
	&=\gamma R^\is(X).
	\end{align*}
	The proof for $R^\sen$ is similar.
\end{proof}

\subsection{Proof of Theorem~\ref{mainthm}}\label{proof}

The proof of Theorem~\ref{mainthm} is preceded by the three lemmata below.

First, let us recall a fundamental result in convex duality. For a function $h\colon \X\to \R\cup\cb{+\infty}$ on a locally convex topological linear space $\X$, we define its epigraph as the set
\[
\epi h \coloneqq  \cb{(x, r)\in \X \times \R \mid h(x)\leq r},
\]
and the conjugate function $h^\ast\colon\X^\ast\to \R\cup\cb{+\infty}$ on the topological dual space $\X^\ast$ by
\[
h^\ast(x^\ast)\coloneqq  \sup_{x\in\X} \of{\ip{x,x^\ast}-h(x)}
\]
for every $x^\ast\in\X^\ast$, where $\ip{\cdot,\cdot}$ is the natural bilinear mapping of the dual pair $(\X^\ast,\X)$. The epigraph $\epi h^\ast$ of $h^\ast$ is defined similarly as a subset of $\X^\ast\times\R$. According to the Fenchel-Moreau biconjugation theorem \cite[Theorem~2.3.3]{zalinescu}, if $h$ is a proper convex lower semicontinuous function, then $h=\of{h^\ast}^\ast$, that is,
\begin{equation}\label{fenchelmoreau}
h(x)=\sup_{x^\ast\in\X^\ast}\of{\ip{x^\ast,x}-h^\ast(x^\ast)}
\end{equation}
for every $x\in\X$. Moreover, if $\eta\colon\X^\ast\to\bar{\R}$ is another function whose closure is $h^\ast$, that is, $\epi h^\ast=\cl\epi\eta$, then we also have $h=\eta^\ast$, that is,
\begin{equation}\label{minorant}
h(x)=\sup_{x^\ast\in\X^\ast}\of{\ip{x^\ast,x}-\eta(x^\ast)}
\end{equation}
for evey $x\in\X$. This is an immediate consequence of \citet[Theorem~2.3.1]{zalinescu}. The next lemma provides a slight variation of \eqref{minorant} that will be useful in the proof of Theorem~\ref{mainthm}.
\begin{lem}\label{minorantnozero}
	Let $h\colon\X\to\R\cup\cb{+\infty}$ be a proper convex lower semicontinuous function and $\eta\colon\X^\ast\to\bar{\R}$ a function whose closure is $h^\ast$. Then, for every $x\in\X$,
	\begin{equation*}
	\sup_{x^\ast\in\X^\ast\sm\cb{0}}\of{\ip{x^\ast,x}-h^\ast(x^\ast)}=\sup_{x^\ast\in\X^\ast\sm\cb{0}}\of{\ip{x^\ast,x}-\eta(x^\ast)}.
	\end{equation*}
\end{lem}
\begin{proof}
	Let $x\in\X$. It is easy to see that
	\[
	\sup_{x^\ast\in\X^\ast\sm\cb{0}}\of{\ip{x^\ast,x}-h^\ast(x^\ast)}=\sup_{x^\ast\in\X^\ast\sm\cb{0}}\ip{(x^\ast,h^\ast(x^\ast)),(x,-1)}=\sup_{(x^\ast,s)\in\epi h^\ast, x^\ast\neq 0}\ip{(x^\ast,s),(x,-1)},
	\]
	where, with a slight abuse of notation, $\ip{(\cdot,\cdot),(\cdot,\cdot)}$ denotes the natural bilinear mapping of the dual pair $(\X^\ast\times\R,\X\times\R)$ of product spaces. We claim that
	\begin{equation}\label{epinotzero}
	\sup_{\substack{(x^\ast,s)\in\epi h^\ast\colon\\ x^\ast\neq 0}}\ip{(x^\ast,s),(x,-1)}=\sup_{\substack{(x^\ast,s)\in\epi \eta\colon\\ x^\ast\neq 0}}\ip{(x^\ast,s),(x,-1)}.
	\end{equation}
	The $\geq$ part is clear since $\epi h^\ast =\cl \epi \eta \supseteq \epi \eta$. To show the $\leq$ part, let $(x^\ast,s)\in\epi h^\ast$ with $x^\ast\neq 0$. So there exists a net $(x^\ast_\theta,s_{\theta}))_{\theta\in\Theta}$ in $\epi \eta$ that converges to $(x^\ast,s)$. ($\Theta$ denotes the directed index set of the net.) Moreover, since $x^\ast\neq 0$, we can have $x^\ast_\theta=0$ only for finitely many $\theta\in\Theta$. Excluding such indices and passing to a subnet, we can assume without loss of generality that $x^\ast_\theta\neq 0$ for every $\theta\in\Theta$. Hence, by the continuity of the bilinear mapping,
	\[
	\sup_{\substack{(x^\ast,s)\in\epi \eta\colon\\ x^\ast\neq 0}}\ip{(x^\ast,s),(x,-1)}\geq \lim_{\theta\in\Theta}\ip{(x^\ast_\theta,s_\theta),(x,-1)}=\ip{(x^\ast,s),(x,-1)}.
	\]
	Since $(x^\ast,s)\in\epi h^\ast$ with $x^\ast\neq 0$ is arbitrary, the $\leq$ part of \eqref{epinotzero} follows. Therefore,
	\begin{align*}
	\sup_{x^\ast\in\X^\ast\sm\cb{0}}\of{\ip{x^\ast,x}-h^\ast(x^\ast)}&=\sup_{\substack{(x^\ast,s)\in\epi \eta\colon\\ x^\ast\neq 0}}\ip{(x^\ast,s),(x,-1)}\\
	&=\sup_{x^\ast\in\X^\ast\sm\cb{0}}\ip{(x^\ast,\eta(x^\ast)),(x,-1)}\\
	&=\sup_{x^\ast\in\X^\ast\sm\cb{0}}\of{\ip{x^\ast,x}-\eta(x^\ast)},
	\end{align*}
	which finishes the proof.
\end{proof}


Consider the function $f\colon L_d^1\to \bar{\R}$ defined by
\begin{equation}\label{f_lambda}
f(U)\coloneqq \inf_{V\in -L^1_{+}}\cb{-\E\sqb{Vg\of{\frac{U}{V}}1_{\cb{V<0}}}+ \rho^\ast\of{V}\mid \E\sqb{V}=-1,\ \Pr\cb{V=0,U\neq 0}=0}
\end{equation}
for $U\in -L_{d,+}^1$, and by $f(U)=+\infty$ for $U\notin-L_{d,+}^1$. 

\begin{lem}\label{conj}
	The function $\of{\rho\circ\Lambda}^\ast$ is the closure of $f$, that is, $\epi\of{\rho\circ\Lambda}^\ast = \cl \epi f$.
\end{lem}

\begin{proof}
	The proof is based on a general conjugation theorem for the composition of an increasing convex function with a convex function, see \citet[Theorem~3.1]{radu}. To that end, let us consider $\Lambda$ as a function on $L_d^\infty$ with values in $L^\infty$, which is a convex function when $L^\infty$ is partially ordered by the cone $-L^\infty_{+}$: $\Lambda(\gamma Z^1+(1-\gamma)Z^2)\in  \gamma \Lambda(Z^1)+(1-\gamma)\Lambda(Z^2)- L_+^\infty$ for every $Z^1,Z^2\in L_d^\infty$, $\gamma\in [0,1]$. Similarly, $\rho$ is increasing with respect to this partial order on $L^\infty$: $Y^1 \in Y^2 - L_+^\infty$ implies $\rho(Y^1) \geq \rho(Y^2)$ for every $Y^1,Y^2\in L^\infty$. Let us define $h_V(Z)=\E\sqb{V\Lambda(Z)}$ for every $V\in L^1$ and $Z\in L_d^\infty$. By \citet[Theorem~3.1]{radu}, using Assumption~\ref{assume}, $\of{\lambda\rho\circ\Lambda}^\ast$ is the closure of the function $\bar{f}$ defined by
	\[
	\bar{f}(U) = \inf_{V\in -L^1_+} \of{h_V^\ast(U)+\rho^\ast(V)}
	\]
	for every $U\in L_d^1$. Let us fix $U\in L_d^1$ and $V\in -L_+^1$. We have
	\begin{align*}
	h_V^\ast(U) &= \sup_{X\in L_d^\infty}\of{\E\sqb{U^{\mathsf{T}}X}-\E\sqb{V\Lambda(X)}}\\
	&= \sup_{X\in L_d^\infty}\E\sqb{U^{\mathsf{T}}X-V\Lambda(X)}\\
	&= \E\sqb{\sup_{x\in \R^d}\of{U^{\mathsf{T}}x-V\Lambda(x)}},
	\end{align*}
	where the last line follows by the general rule for the optimization of integral functionals; see \citet[Theorem~14.60]{rockafellar2}. For $u\in\R^d,v\in(-\infty,0]$, note that 
	\[
	\sup_{x\in\R^d}\of{u^\mathsf{T}x-v\Lambda(x)}=\begin{cases}0&\text{if }v=0, u=0,\\
	+\infty&\text{if }v=0, u\neq 0,\\
	-vg\of{\frac{u}{v}}&\text{if }v<0.\end{cases}
	\]
	Recalling that $g(z)=+\infty$ for every $z\notin\R^d_+$, we may write
	\[
	\sup_{x\in\R^d}\of{u^\mathsf{T}x-v\Lambda(x)}=\begin{cases}0&\text{if }v=0, u=0,\\
	+\infty&\text{if }v=0, u\neq 0,\\
	+\infty&\text{if }v<0, u\notin-\R^d_+,\\
	-vg\of{\frac{u}{v}}&\text{if }v<0, u\in-\R^d_+.\end{cases}
	\]
	In particular, if $u\notin-\R^d_+$, then
	\[
	\sup_{x\in\R^d}\of{u^\mathsf{T}x-v\Lambda(x)}=+\infty
	\]
	for every $v\leq 0$. Hence, if $U\notin-L_{d,+}^1$, then $h_V^\ast(U)=+\infty$ for every $V\in -L_+^1$ so that $\bar{f}(U)=f(U)=+\infty$. Let us assume that $U\in -L^1_{d,+}$ and take $V\in-L_+^1$. Then, 
	\[
	h_V^\ast(U)=\begin{cases}+\infty&\text{if }\Pr\cb{V=0,U\neq 0}>0,\\
	-\E\sqb{Vg\of{\frac{U}{V}}1_{\cb{V<0}}}&\text{if }\Pr\cb{V=0,U\neq 0}=0.\end{cases}
	\]
	On the other hand, using the monotonicity and translativity of $\rho$, it can be checked that $\rho^\ast \of{V}<+\infty$ implies $\E\sqb{V}=-1$; see \citet[Remark~4.18]{fs:sf}, for instance. Hence, $\bar{f}(U)=f(U)$ when $U\in-L_{d,+}^1$; see \eqref{f_lambda}. Therefore, the functions $f$ and $\bar{f}$ coincide, and the result follows.
	
\end{proof}

Before the proof of Theorem~\ref{mainthm}, we provide a lemma of independent interest. It should be a known result, the proof is included for completeness.
\begin{lem}\label{measures}
	Let $\mu_1,\mu_2$ be two finite measures on $(\Omega,\F)$ such that $\mu_1\ll \Pr$ and $\mu_2\ll\Pr$. Then, $\mu_1\ll \mu_2$ if and only if
	\begin{equation}\label{rncond}
	\Pr\cb{\frac{d\mu_2}{d\Pr}=0,\ \frac{d\mu_1}{d\Pr}>0}=0.
	\end{equation}
\end{lem}
\begin{proof}
	Suppose that $\mu_1\ll\mu_2$. Then, by a corollary of Radon-Nikodym theorem, see \citet[Exercise~A.2.1]{fs:sf}, for instance, we may write
	\[
	\frac{d\mu_1}{d\Pr}=\frac{d\mu_1}{d\mu_2}\cdot\frac{d\mu_2}{d\Pr},
	\]
	where the equality is understood in the $\Pr$-almost sure sense. Hence, with $\Pr$-probability one, $\frac{d\mu_2}{d\Pr}=0$ implies that $\frac{d\mu_1}{d\Pr}=0$ so that \eqref{rncond} holds. Conversely, suppose that \eqref{rncond} holds. Let $A\in\F$ be an event such that
	\[
	\mu_2(A)=\E\sqb{1_A\frac{d\mu_2}{d\Pr}}=0.
	\]
	Hence, $1_A\frac{d\mu_2}{d\Pr}=0$ $\Pr$-almost surely so that $\Pr(A\cap\{\frac{d\mu_2}{d\Pr}>0\})=0$. This and \eqref{rncond} imply that
	\[
	\mu_1(A)=\E\sqb{1_A\frac{d\mu_1}{d\Pr}}=\E\sqb{1_A\frac{d\mu_1}{d\Pr}1_{\cb{\frac{d\mu_2}{d\Pr}>0}}}+\E\sqb{1_A\frac{d\mu_1}{d\Pr}1_{\cb{\frac{d\mu_2}{d\Pr}=0}}}=0.
	\]
	This shows that $\mu_1\ll\mu_2$.
\end{proof}

\begin{proof}[Proof of Theorem~\ref{mainthm}]\text{}
	\begin{enumerate}[1.]
		\item By the arguments in the proof of Proposition~\ref{properties}, it follows that $\rho\circ\Lambda$ is a proper convex weak$^\ast$ lower semicontinuous function on $L_d^\infty$. It is also a decreasing function so that $\of{\rho\circ\Lambda}^\ast(U)=+\infty$ for $U\notin - L_{d,+}^1$. Let $X\in L_d^\infty$. By the Fenchel-Moreau biconjugation theorem,
		\begin{align}\label{rholam}
		\rho\of{\Lambda(X)} =\sup_{U\in -L_{d,+}^1}\of{\E\sqb{U^{\mathsf{T}}X}-\of{ \rho \circ \Lambda}^\ast(U)}.
		\end{align}
		Moreover, we can indeed exclude $U\equiv 0$ in this computation and write
		\begin{align}\label{nozero}
		\rho\of{\Lambda(X)} =\sup_{U\in -L_{d,+}^1\sm\cb{0}}\of{\E\sqb{U^{\mathsf{T}}X}-\of{ \rho \circ \Lambda}^\ast(U)}.
		\end{align}
		To see this, we first claim that there exists $\bar U\in -L_{d,+}^1\sm\cb{0}$ with $\of{\rho\circ \Lambda}^\ast(\bar U)\in\R$. Suppose otherwise. Since $(\rho\circ \Lambda)^\ast$ is a proper function as the conjugate of a proper function, we must have $\of{\rho\circ\Lambda}^\ast(\bar U)<+\infty$ if and only if $\bar U=0$. By \eqref{rholam}, this would imply $\rho(\Lambda(X))=-\of{\rho\circ\Lambda}^\ast(0)=\inf_{Z\in L_d^\infty}\rho(\Lambda(Z))$ for every $X\in L_d^\infty$ so that $\rho\circ\Lambda$ is a constant function. This is a contradiction to the assumption that $\Lambda$ is non-constant; see Section~\ref{general}. Hence, the claim holds. Let us define $U^n \coloneqq \frac{1}{n}\bar U \in -L_{d,+}^1\sm\cb{0}$ for each $n\in\mathbb{N}$ so that $(U^n)_{n\in\mathbb{N}}$ converges to $0$ in $L_d^1$. Using the concavity of the function $U\mapsto \zeta(U)\coloneqq \E\sqb{U^{\mathsf{T}}X}-\of{ \rho \circ \Lambda}^\ast(U)$, it holds
		\begin{align}\label{convergence}
		\sup_{U\in -L_{d,+}^1\sm\cb{0}}\zeta(U)\geq \zeta(U^n)\geq \frac{1}{n}\zeta(\bar U)+\of{1-\frac{1}{n}}\zeta(0)
		\end{align}
		for every $n\in\mathbb{N}$. Since we have $-\infty<\zeta(\bar U)\leq \rho\of{\Lambda(X)}<+\infty$ by \eqref{rholam} and the finite-valuedness of $\rho$, we obtain
		\[
		\sup_{U\in -L_{d,+}^1\sm\cb{0}}\zeta(U) \geq \zeta(0)
		\]
		by passing to the limit in \eqref{convergence} as $n\rightarrow\infty$. It follows that \eqref{nozero} holds.
		
		Next, by \eqref{nozero}, Lemma~\ref{minorantnozero} and Lemma~\ref{conj}, we obtain
		\begin{align*}
		\rho\of{\Lambda(X)}  & =\sup_{U\in -L_{d,+}^1\sm\cb{0}}\of{\E\sqb{U^{\mathsf{T}}X}-f(U)}\\
		&= \sup_{U\in -L_{d,+}^1\sm\cb{0}}\sup_{\substack{V\in - L_{+}^1 \colon \E\sqb{V}=-1,\\ \Pr\cb{V=0,U\neq 0}=0}}\of{\E\sqb{U^{\mathsf{T}}X} +\E\sqb{Vg\of{\frac{U}{V}}1_{\cb{V<0}}}-\rho^\ast(V)}.
		\end{align*}
		
		Note that every $U\in - L_{d,+}^1\sm\cb{0}$ can be identified by a vector probability measure $\Q\in \M_d(\Pr)$ and a weight vector $w\in \R^d_+\sm\cb{0}$ (and vice versa) by the relationship
		\[
		w_i = -\E\sqb{U_i},\quad w_i\frac{d\Q_i}{d\Pr} = -U_i,\quad i\in\cb{1,\ldots,d}.
		\]
		(Note that $\Q_i$ is defined arbitrarily when $w_i=0$.) Similarly, every $V\in -L_{+}^1$ with $\E\sqb{V}=-1$ can be identified by a probability measure $\S\in\M(\Pr)$ by setting
		\begin{align}\label{dSdP}
		\frac{d\S}{d\Pr}=-V.
		\end{align}
		In this case, $\rho^\ast(V)=\a(\S)$; see \citet[Remark~4.18]{fs:sf}, for instance. With these changes of variables, the condition $\Pr\cb{V=0,U\neq0}=0$ becomes
		\[
		\Pr\cb{\frac{d\S}{d\Pr}=0,\ w\cdot\frac{d\Q}{d\Pr}\neq 0}=\Pr\of{\bigcup_{i=1}^d\cb{\frac{d\S}{d\Pr}=0,\ \frac{d(w_i\Q_i)}{d\Pr}>0}}=0,
		\]
		which is equivalent to having
		\[
		\Pr\cb{\frac{d\S}{d\Pr}=0,\frac{d(w_i\Q_i)}{d\Pr}> 0}=0
		\]
		for every $i\in\cb{1,\ldots,d}$. By Lemma~\ref{measures}, we conclude that $\Pr\cb{V=0,U\neq0}=0$ is equivalent to that $w_i\Q_i\ll\S$ for every $i\in\cb{1,\ldots,d}$. In particular, we may write
		\[
		\frac{U}{V}=\frac{w\cdot\frac{d\Q}{d\Pr}}{\frac{d\S}{d\Pr}}=\of{\frac{d(w_1\Q_1)}{d\S},\ldots,\frac{d(w_d\Q_d)}{d\S}}=w\cdot\frac{d\Q}{d\S}.
		\]
		(When $w_i=0$ for some $i\in\cb{1,\ldots,d}$, $\Q_i$ is defined arbitrarily and $w_i\frac{d\Q_i}{d\S}=0$ is understood.) As a result, we may write 
		\begin{align}
		\rho\of{\Lambda(X)}  
		&= \sup_{\substack{\Q\in\M_d(\Pr),\\ w\in\R^d_+\sm\cb{0}}}\sup_{\substack{ \S\in\M(\Pr)\colon\\ \forall i\colon w_i\Q_i\ll\S}}\of{w^{\mathsf{T}}\E^\Q\sqb{-X} -\E^\S\sqb{g\of{w\cdot \frac{d\Q}{d\S}}1_{\cb{\frac{d\S}{d\Pr}>0}}}-\a(\S)}\notag \\
		&=\sup_{\substack{\Q\in\M_d(\Pr),\\ w\in\R^d_+\sm\cb{0}}}\of{w^{\mathsf{T}}\E^\Q\sqb{-X}-\inf_{\substack{ \S\in\M(\Pr)\colon\\ \forall i\colon w_i\Q_i\ll\S}}\of{ \E^\S\sqb{g\of{w\cdot \frac{d\Q}{d\S}}}+\a(\S)}}\notag \\
		&=\sup_{\substack{\Q\in\M_d(\Pr),\\ w\in\R^d_+\sm\cb{0}}}\of{w^{\mathsf{T}}\E^\Q\sqb{-X}-\a^\s(\Q,w)}.\label{compdual}
		\end{align}
		
		Finally, recalling \eqref{ins-rep}, we obtain
		\begin{align*}
		R^\is(X) &= \cb{z\in\R^d \mid \mathbf{1}^{\mathsf{T}}z \geq \rho\of{\Lambda(X)}}\\
		&=\cb{z\in \R^d \mid \mathbf{1}^{\mathsf{T}}z\geq \sup_{\Q\in \M_d(\Pr), w\in\R^d_+\sm\cb{0}}\of{w^{\mathsf{T}}\E^{\Q}\sqb{-X}-\a^{\sys}(\Q,w)}}\\
		&=\bigcap_{\Q\in \M_d(\Pr), w\in\R^d_+\sm\cb{0}}\cb{z\in \R^d \mid \mathbf{1}^{\mathsf{T}}z\geq w^{\mathsf{T}}\E^{\Q}\sqb{-X}-\a^{\sys}(\Q,w)},
		\end{align*}
		which finishes the proof of the dual representation for $R^\is$.
		
		\item Using \eqref{compdual}, we obtain
		\begin{align*}
		R^{\sen}(X)
		&=\cb{z\in\R^d\mid\rho\of{\Lambda\of{X+z}}\leq 0}\\
		&=\Big\{z\in\R^d\mid \sup_{\substack{\Q\in\M_d(\Pr),\\ w\in\R^d_+\sm\cb{0}}}\of{w^{\mathsf{T}}\E^\Q\sqb{-\of{X+z}}-\a^\s(\Q,w)}\leq 0\Big\}\\
		&=\bigcap_{\Q\in\M_d(\Pr),w\in\R^d_+\sm\cb{0}}\E^\Q\sqb{-X}+\cb{z\in\R^d\mid w^{\mathsf{T}}z\geq  -\a^\s(\Q,w)},
		\end{align*}
		which finishes the proof of the dual representation for $R^\sen$.
	\end{enumerate}
\end{proof}

Recall that, for a linear space $\X$ and a set $A\subseteq \X$, the convex-analytic indicator function $I_A\colon\X\to\R\cup\cb{+\infty}$ of $A$ is defined by $I_A(x)=0$ for $x\in A$ and $I_A(x)=+\infty$ for $x\in \X\sm A$.

\begin{proof}[Proof of Proposition~\ref{scalarizations}]
	The dual representation of $\rho^\is$ follows as a direct consequence of the first part of Theorem~\ref{mainthm} since $R^\is$ is a halfspace-valued function.
	
	We prove the dual representation of $\rho^\sen_w$ for a given weight vector $w\in\R^d_+\sm\cb{0}$. Using the properties of the set-valued risk measure $R^\sen$ in Proposition~\ref{properties}, it can be checked that $\rho^\sen_{w}$ is a decreasing proper convex function. Moreover, we assume that it is weak* lower semicontinuous. Let us compute its conjugate function $(\rho^\sen_{w})^\ast$ at $U\in L_d^1$. Note that, as $\rho^\sen_{w}$ is a decreasing function, it holds $(\rho^\sen_{w})^\ast (U)=+\infty$ if $U\notin -L_{d,+}^1$. Let $U\in-L_{d,+}^1$. Since
	\begin{align*}
	\rho^\sen_w(X)=\inf_{z\in\R^d}\cb{w^{\mathsf{T}}z\mid \Lambda(X+z)\in\A}=\inf_{z\in\R^d}\of{w^{\mathsf{T}}z+I_{\A}(\Lambda(X+z))},
	\end{align*}
	we have
	\begin{align*}
	(\rho^\sen_w)^\ast(U)&=\sup_{\bar{X}\in L_d^\infty}\of{\E\sqb{U^{\mathsf{T}}\bar{X}}-\rho^\sen_w(\bar{X})}\\
	&=\sup_{z\in\R^d}\sup_{\bar{X}\in L_d^\infty}\of{\E\sqb{U^{\mathsf{T}}\bar{X}}-w^{\mathsf{T}}z-I_{\A}(\Lambda(\bar{X}+z))}\\
	&=\sup_{z\in\R^d}\sup_{X\in L_d^\infty}\of{\E\sqb{U^{\mathsf{T}}(X-z)}-w^{\mathsf{T}}z-I_{\A}(\Lambda(X))}\\
	&=\sup_{X\in L_d^\infty}\of{\E\sqb{U^{\mathsf{T}}X}-I_{\A}\circ\Lambda(X)}+\sup_{z\in\R^d}(\E\sqb{-U}-w)^{\mathsf{T}}z.
	\end{align*}
	Hence, $(\rho^\sen_{w})^\ast (U) = +\infty$ if $\E\sqb{U}\neq -w$. On the other hand, if $\E\sqb{U}=-w$, then
	\begin{align*}
	(\rho^\sen_{w})^\ast (U)=  \sup_{X\in L_d^\infty}\of{\E\sqb{U^{\mathsf{T}}X}-I_{\A}\circ\Lambda(X)}= \of{I_{\A}\circ\Lambda}^\ast (U).
	\end{align*}
	So
	\[
	(\rho^\sen_{w})^\ast (U)=  \of{I_{\A}\circ\Lambda}^\ast(U)+I_{\cb{\E\sqb{\cdot}+w=0}}(U).
	\]
	We calculate the conjugate $\of{I_{\A}\circ\Lambda}^\ast$ as the closure of a function following a similar route as in the proof of Lemma~\ref{conj} (for $(\rho\circ\Lambda)^\ast$ there). Using $h_V(Z)=\E\sqb{V\Lambda(Z)}$ for $V\in L^1, Z\in L_d^\infty$ as before, by \citet[Theorem~3.1]{radu}, $\of{I_{\A}\circ\Lambda}^\ast$ is the closure of the function $\bar{m}$ defined by
	\[
	\bar{m}(U)=\inf_{V\in -L_+^1}\of{h_V^\ast(U)+(I_\A)^\ast(V)}
	\]
	for every $U\in L_d^1$. As in the proof of Lemma~\ref{conj}, we have
	\[
	h_V^\ast(U)=\begin{cases}+\infty&\text{if }\Pr\cb{V=0,U\neq 0}>0,\\
	-\E\sqb{Vg\of{\frac{U}{V}}1_{\cb{V<0}}}&\text{if }\Pr\cb{V=0,U\neq 0}=0.\end{cases}
	\]
	In particular, similar to the proof of Lemma~\ref{conj}, it can be checked that $\bar{m}(U)=+\infty$ if $U\notin -L_{d,+}^1$. 
	
	
	
	Let $V\in -L^1_+$. If $V\equiv 0$, then $(I_\A)^\ast(V)=0$. Suppose that $\Pr\cb{V<0}>0$. We have $\E\sqb{-V}>0$ and
	\[
	(I_\A)^\ast(V)=\sup_{Y\in \A}\E\sqb{VY}=\E\sqb{-V}\sup_{Y\in\A}\E\sqb{\frac{V}{\E\sqb{-V}}Y}=\E\sqb{-V}\rho^\ast\of{\frac{V}{\E\sqb{-V}}}.
	\]
	
	Consequently, for every $U\in -L_{d,+}^1$, we have
	\begin{align*}
	\bar{m}(U)=\inf_{V\in -L_+^1}\cb{-\E\sqb{Vg\of{\frac{U}{V}}1_{\cb{V<0}}}+\E\sqb{-V}\rho^\ast\of{\frac{V}{\E\sqb{-V}}}\mid \Pr\cb{V=0,U\neq 0}=0},
	\end{align*}
	where $\E\sqb{-V}\rho^\ast(\frac{V}{\E\sqb{-V}})=0$ is understood when $V\equiv 0$. We also have $\bar{m}(U)=+\infty$ for every $U\notin -L_{d,+}^1$; that is, $\bar{m}(U)=m(-U)$ for every $U\in L_d^1$, where $m$ is the function defined in \eqref{mbar}.
	
	We have
	\[
	(\rho_w^\sen)^\ast(U)=(\cl \bar{m})(U)+I_{\cb{\E\sqb{\cdot}+w=0}}(U)
	\]
	for every $U\in L_d^1$. Since we assume that $\rho^\sen_w$ is weak* lower semicontinuous, by Fenchel-Moreau theorem,
	\[
	\rho^\sen_w(X)=\sup_{\substack{U\in -L_{d,+}^1\colon\\ \E[U]=-w}}\of{\E\sqb{U^{\mathsf{T}}X}-(\cl \bar{m})(U)}.
	\]
	Note that every $U\in -L_{d,+}^1$ with $\E\sqb{U}=-w$ can be identified by a vector probability measure $\Q\in \M_d(\Pr)$ (and vice versa) by the relationship
	\[
	w_i\frac{d\Q_i}{d\Pr}= -U_i, \quad i\in\cb{1,\ldots,d}
	\]
	so that
	\begin{align}\label{mbarrep}
	\rho^\sen_{w}(X) &=\sup_{\Q\in \M_d(\Pr)}\of{w^{\mathsf{T}}\E^\Q\sqb{-X}-(\cl \bar{m}) \of{-w\cdot\frac{d\Q}{d\Pr}}}.
	\end{align}
	Since $\bar{m}(-U)=m(U)$ for every $U\in L_d^1$, it follows that $(\cl \bar{m})(-U)=(\cl m)(U)$ for every $U\in L_d^1$. This and \eqref{mbarrep} finish the proof of \eqref{mrep}:
	\[
	\rho^\sen_{w}(X) =\sup_{\Q\in \M_d(\Pr)}\of{w^{\mathsf{T}}\E^\Q\sqb{-X}-(\cl m) \of{w\cdot\frac{d\Q}{d\Pr}}}.
	\]
	
	Next, suppose further that $m$ is a lower semicontinuous function so that
	\[
	(\cl m)(U)=m(U)
	\]
	for every $U\in L_d^1$. Let $U\in -L_{d,+}^1$ with $\E\sqb{U}=-w$. Then, for $V\equiv 0$, we have $\Pr\cb{V=0,U\neq 0}=\Pr\cb{U\neq 0}>0$ so that $h_V^\ast(U)=+\infty$. Hence, we may write
	\begin{equation}\label{vnonzero}
	m(-U)=\bar{m}(U)=\inf_{V\in -L_+^1\sm\cb{0}}\of{h_V^\ast(U)+(I_\A)^\ast(V)}
	\end{equation}
	in this case. In particular, for $\Q\in\M_d(\Pr)$, thanks to $\eqref{vnonzero}$, we have 
	\begin{align*}
	&(\cl m) \of{w\cdot\frac{d\Q}{d\Pr}}=m\of{w\cdot\frac{d\Q}{d\Pr}}\\
	&=\inf_{V\in -L_+^1\sm\cb{0}}\cb{-\E\sqb{Vg\of{-\frac{w\cdot\frac{d\Q}{d\Pr}}{V}}1_{\cb{V<0}}}+\E\sqb{-V}\rho^\ast\of{\frac{V}{\E\sqb{-V}}}\mid \Pr\cb{V=0,w\cdot\frac{d\Q}{d\Pr}\neq 0}=0}.
	\end{align*}
	Following a similar variable transform as in the proof of Theorem~\ref{mainthm}, we may write every $V\in -L_{+}^1\sm\cb{0}$ as
	\[
	-V=\lambda\frac{d\S}{d\Pr},
	\]
	where $\lambda =\E\sqb{-V}>0$ and $\S\in\M(\Pr)$ is such that $\frac{d\S}{d\Pr}=-\frac{V}{\lambda}$ (and vice versa). This calculation yields
	\begin{align*}
	m\of{w\cdot\frac{d\Q}{d\Pr}}&=\inf_{\lambda>0}\inf_{\substack{\S\in\M(\Pr)\colon \\
			\forall i\colon w_i\Q_i\ll \S} }\of{\lambda\E^\S\sqb{g\of{\frac{w}{\lambda}\cdot\frac{d\Q}{d\S}}1_{\cb{\frac{d\S}{d\Pr}>0}}}+\lambda\rho^\ast\of{\frac{d\S}{d\Pr}}}\\
	&=\inf_{\lambda>0}\inf_{\substack{\S\in\M(\Pr)\colon \\
			\forall i\colon w_i\Q_i\ll \S} }\of{\lambda\E^\S\sqb{g\of{\frac{w}{\lambda}\cdot\frac{d\Q}{d\S}}}+\lambda\a(\S)}.
	\end{align*}
	Using \eqref{mbarrep}, we have
	\begin{align*}
	\rho^\sen_{w}(X)
	&=\sup_{\Q\in \M_d(\Pr)}\of{w^{\mathsf{T}}\E^\Q\sqb{-X}-m \of{w\cdot\frac{d\Q}{d\Pr}}}\\
	&= \sup_{\Q\in \M_d(\Pr)}  \sup_{\lambda>0} \of{w^{\mathsf{T}}\E^\Q\sqb{-X} -\lambda \inf_{\substack{\S\in\M(\Pr)\colon\\ \forall i\colon w_i\Q_i\ll\S}}\of{\E^\S\sqb{g\of{\frac{w}{\lambda}\cdot\frac{d\Q}{d\S}}} +\alpha(\S)}}\\
	&= \sup_{\Q\in \M_d(\Pr)}\sup_{\lambda>0}   \of{w^{\mathsf{T}}\E^\Q\sqb{-X} -\lambda \a^\sys(\Q,\frac{w}{\lambda})}\\
	&= \sup_{\Q\in \M_d(\Pr)}  \of{w^{\mathsf{T}}\E^\Q\sqb{-X} -\tilde{\a}^\s (\Q,w)},
	\end{align*}
	which yields \eqref{scsimple}.
	
	Finally, suppose that there exist $\hat{X}\in L_d^\infty$ and a weak* neighborhood $A$ of $\Lambda(\hat{X})$ with $A\subseteq\A$. Then, $I_\A$ is bounded from above (by zero) on $A$. By \citet[Theorem~2.2.9]{zalinescu}, $I_\A$ is weak* continuous at $\Lambda(\hat{X})$. By the stronger conjugation result \citet[Theorem~2.8.10(iii)]{zalinescu} (together with the weaker one \citet[Theorem~3.1]{radu}), we precisely have $(I_\A\circ\Lambda)^\ast=\bar{m}=\cl(\bar{m})$. Hence, $m$ is lower semicontinuous and \eqref{scsimple} holds in this special case.
\end{proof}

\small 
\section*{Acknowledgments}
This material is based upon work supported by the National Science Foundation under Grant No.~1321794 and the OeNB anniversary fund, project number 17793. Part of the manuscript was written when the first author visited Vienna University of Economics and Business. The authors would like to thank an anonymous referee for useful comments and suggestions that helped improving the manuscript. The first author would like to thank Fabio Bellini, Zachary Feinstein and Daniel Ocone for fruitful discussions. The authors would like to thank Cosimo-Andrea Munari and Maria Arduca for pointing out an issue in an earlier version of the paper, as well as Alexander Smirnow and Jana Hlavinova for pointing out a simplification in the proof of the second part of Theorem~\ref{mainthm}.

\bibliographystyle{named}

\begin{thebibliography}{99}
	\bibitem[Amini \emph{et~al.}(2015)]{amini}
	H.~Amini, D.~Filipovic, A.~Minca, \emph{Systemic risk and central clearing counterparty design}, Swiss Finance Institute Research Paper No. 13--34, SSRN e-prints, \href{http://papers.ssrn.com/sol3/papers.cfm?abstract_id=2275376}{2275376}, 2015.
	
	\bibitem[Ararat \emph{et~al.}(2017)]{sdrm}
	\c{C}.~Ararat, A.~H.~Hamel, B.~Rudloff, \emph{Set-valued shortfall and divergence risk measures}, International Journal of Theoretical and Applied Finance, \textbf{20}(5): 1750026 (48 pages), 2017. 
	
	\bibitem[Armenti \emph{et~al.}(2018)]{samuel}
	Y.~Armenti, S.~Cr\'{e}pey, S.~Drapeau, A.~Papapantoleon, \emph{Multivariate shortfall risk allocation}, SIAM Journal on Financial Mathematics, \textbf{9}(1): 90-126, 2018.
	
	\bibitem[Artzner \emph{et~al.}(1999)]{artzner}
	P.~Artzner, F.~Delbaen, J.-M.~Eber, D.~Heath, \emph{Coherent measures of risk}, Mathematical Finance, \textbf{9}(3): 203–-228, 1999.
	
	\bibitem[Biagini \emph{et~al.}(2019a)]{fouque}
	F.~Biagini, J.-P.~Fouque, M.~Fritelli, T.~Meyer-Brandis, \emph{A unified approach to systemic risk measures via acceptance sets}, Mathematical Finance, \textbf{29}(1): 329--367, 2019.
	
	\bibitem[Biagini \emph{et~al.}(2019b)]{fairness}
	F.~Biagini, J.-P.~Fouque, M.~Fritelli, T.~Meyer-Brandis, \emph{On fairness of systemic risk measures}, arXiv e-prints, \href{http://arxiv.org/abs/1803.09898}{1803.09898}, 2019.
	
	\bibitem[Bo\c{t} \emph{et~al.}(2009)]{radu}
	R.~I.~Bo\c{t}, S.-M.~Grad, G.~ Wanka, \emph{Generalized Moreau-Rockafellar results for composed convex functions}, Optimization, \textbf{58}(7): 917--933, 2009.
	
	\bibitem[Brunnermeier, Cheridito(2019)]{dito}
	M.~K.~Brunnermeier, P.~Cheridito, \emph{Measuring and allocating systemic risk}, Risks \textbf{7}(2), 46: 1--19, 2019.
	
	\bibitem[Chen \emph{et~al.}(2013)]{cim}
	C.~Chen, G.~Iyengar, C.~C.~Moallemi, \emph{An axiomatic approach to systemic risk}, Management Science, \textbf{59}(6), 1373-1388, 2013.
	
	\bibitem[Cifuentes \emph{et~al.}(2005)]{clearing}
	R.~Cifuentes, H.~S.~Shin, G.~Ferrucci, \emph{Liquidity risk and contagion}, Journal of the European Economic Association, \textbf{3}(2-3): 556–-566, 2005.
	
	
	\bibitem[Eisenberg, Noe(2001)]{eisnoe}
	L.~Eisenberg, T.~H.~Noe, \emph{Systemic risk in financial systems}, Management Science, \textbf{47}(2): 236--249, 2001.
	
	\bibitem[Farkas \emph{et~al.}(2015)]{multipleeligible}
	W.~Farkas, P.~Koch-Medina, C.~Munari, \emph{Measuring risk with multiple eligible assets}, Mathematics and Financial Economics, \textbf{9}(1): 3--27, 2015.
	
	\bibitem[Feinstein \emph{et~al.}(2017)]{syst}
	Z.~Feinstein, B.~Rudloff, S.~Weber, \emph{Measures of systemic risk}, SIAM Journal on Financial Mathematics, \textbf{8}(1): 672-–708, 2017.
	
	\bibitem[F\"{o}llmer, Schied(2002)]{scalarshortfall}
	H.~F\"{o}llmer, A.~Schied, \emph{Convex measures of risk and trading constraints}, Finance and Stochastics, \textbf{6}(4): 429-447, 2002.
	
	\bibitem[F\"{o}llmer, Schied(2011)]{fs:sf}
	H.~F\"{o}llmer, A.~Schied, \emph{Stochastic finance: an introduction in discrete time}, De Gruyter Textbook Series, third edition, 2011.
	
	\bibitem[Hamel, Heyde(2010)]{hh:duality}
	A.~H.~Hamel and F.~Heyde, \emph{Duality for set-valued measures of risk}, SIAM Journal on Financial Mathematics, \textbf{1}(1): 66--95, 2010.
	
	\bibitem[Hamel \emph{et al.}(2015)]{setoptsurv}
	A.~H.~Hamel, F.~Heyde, A.~L\"{o}hne, B.~Rudloff, C.~Schrage, \emph{Set optimization - a rather short introduction}, in: A.~H.~Hamel, F.~Heyde, A.~L\"{o}hne, B.~Rudloff, C.~Schrage (eds.), Set optimization and applications - the state of the art. From set relations to set-valued risk measures, 65--141, Springer-Verlag Berlin, 2015.
	
	\bibitem[Harris, Ross(1955)]{harris}
	T.~E.~Harris and F.~S.~Ross, \emph{Fundamentals of a method for evaluating rail net capacities}, Research Memorandum RM-1573, The RAND Corporation, Santa Monica, California, 1955.
	
	\bibitem[Hoffmann \emph{et~al.}(2016)]{hoffmann}
	H.~Hoffmann, T.~Meyer-Brandis, G.~Svindland, \emph{Risk-consistent conditional systemic risk measures}, Stochastic Processes and their Applications, \textbf{126}(7): 2014--2037, 2016.
	
	\bibitem[Kabanov \emph{et~al.}(2017)]{kabanovsurvey}
	Y.~Kabanov, R.~Mokbel, K.~El~Bitar, \emph{Clearing in financial networks}, Theory of Probability and Its Applications, \textbf{62}(2): 311--344, 2017.
	
	\bibitem[Kromer \emph{et~al.}(2016)]{kromer}
	E.~Kromer, L.~Overbeck, K.~Zilch, \emph{Systemic risk measures over general measurable spaces}, Mathematical Methods of Operations Research, \textbf{84}(2): 323--357, 2016.
	
	\bibitem[L\"{o}hne(2011)]{lohnebook}
	A.~L\"{o}hne, \emph{Vector optimization with infimum and supremum}, Springer, 2011.
	
	\bibitem[L\"{o}hne \emph{et~al.}(2014)]{firdevs}
	A.~L\"{o}hne, B.~Rudloff, F.~Ulus, \emph{Primal and dual approximation algorithms for convex vector optimization problems}, Journal of Global Optimization, \textbf{60}(4): 713--736, 2014.
	
	\bibitem[Rockafellar(1970)]{rockafellar}
	R.~T.~Rockafellar, \emph{Convex analysis}, Princeton University Press, 1970.
	
	\bibitem[Rockafellar, Wets(2010)]{rockafellar2}
	R.~T.~Rockafellar, R.~J-B~Wets, \emph{Variational analysis}, Grundlehren der mathematischen Wissenschaften \textbf{317}, Springer-Verlag Berlin Heidelberg, 1998, corrected third printing 2010.
	
	\bibitem[Rogers, Veraart(2013)]{rogersveraart}
	L.~C.~G.~Rogers, L.~A.~M.~Veraart, \emph{Failure and rescue in an interbank network}, Management Science, \textbf{59}(4): 882--898, 2013.
	
	\bibitem[Schrijver(2002)]{maxflowhist}
	A.~Schrijver, \emph{On the history of the transportation and maximum flow problems},
	Mathematical Programming, Series B, \textbf{91}(3): 437–-445, 2002.
	
	\bibitem[Zalinescu(2002)]{zalinescu}
	C.~Zalinescu, \emph{Convex analysis in general vector spaces}, World Scientific, 2002.
\end{thebibliography}

\end{document}